\documentclass{amsart}
\usepackage[T1]{fontenc}
\usepackage{amsmath,a4wide}
\usepackage{amssymb}
\usepackage{amscd}
\usepackage{epsfig}
\usepackage[T1]{fontenc}
\usepackage{epsfig}
\usepackage{amsmath}
\usepackage{braket}
\usepackage{graphicx}
\usepackage{subfigure}
\usepackage{mathrsfs}
\usepackage{epstopdf}
\usepackage{color}
\usepackage{booktabs}
\usepackage{ifpdf}
\usepackage{cite}
\usepackage{amsthm}
\usepackage{amsmath}
\usepackage{dsfont}
\definecolor{myblue}{rgb}{0.2,0,0.9}
\definecolor{blue-violet}{rgb}{0.54, 0.17, 0.89}
\usepackage{enumitem}
\setlist[enumerate]{label={\upshape(\roman*)}}
\usepackage{algorithm,algorithmic}

\usepackage{pgfplots}
\usepgfplotslibrary{groupplots}
\pgfplotsset{compat=1.12}

\usepackage{mathtools}
\mathtoolsset{showonlyrefs}

\usepackage{tikz}
\usetikzlibrary{arrows,intersections}
\usepackage{boldline}
\usepackage{multirow}
\usepackage[margin=1.32in]{geometry}

\usepackage{varioref}
\usepackage{xr-hyper}
\usepackage{hyperref}
\hypersetup{
colorlinks,linkcolor=blue,citecolor=blue,filecolor=red,urlcolor=blue}

\def\HS{\mathrm{HS}}
\definecolor{myblue}{rgb}{0.2,0,0.9}
\definecolor{blue_violet}{rgb}{0.54, 0.17, 0.89}
\definecolor{darkgreen}{rgb}{0,0.35,0}

\makeatletter
\DeclareRobustCommand*\cal{\@fontswitch\relax\mathcal}
\makeatother

\makeatletter
\newcommand{\labeltext}[3][]{%
\@bsphack%
\csname phantomsection\endcsname
\def\tst{#1}%
\def\labelmarkup{\emph}
\def\refmarkup{}%
\ifx\tst\empty\def\@currentlabel{\refmarkup{#2}}{\label{#3}}%
\else\def\@currentlabel{\refmarkup{#1}}{\label{#3}}\fi%
\@esphack%
\labelmarkup{#2}
}
\makeatother

\newtheorem{theorem}{Theorem}[section]
\newtheorem{proposition}[theorem]{Proposition}

\newtheorem{corollary}[theorem]{Corollary}
\newtheorem{lemma}[theorem]{Lemma}
\numberwithin{equation}{section}

\theoremstyle{definition}
\newtheorem{remark}[theorem]{Remark}
\newtheorem{example}[theorem]{Example}
\newtheorem{definition}[theorem]{Definition}

\newcommand\quotientone[2]{{^{#1}}/{_{#2}}}

\DeclareMathOperator*{\argmin}{arg\,min}

\DeclareMathOperator{\trace}{tr}

\DeclareMathOperator{\Pol}{Pol}

\DeclareMathOperator{\range}{rg}
\DeclareMathOperator{\re}{Re}
\DeclareMathOperator{\im}{Im}
\DeclareMathOperator*{\linspan}{span}

\DeclareMathOperator{\CNOT}{CNOT}
\DeclareMathOperator{\SWAP}{SWAP}
\DeclareMathOperator{\stab}{stab}
\DeclareMathOperator{\dom}{dom}
\DeclareMathOperator{\md}{mod}

\usepackage{scalerel,stackengine}
\stackMath
\newcommand\reallywidehat[1]{%
	\savestack{\tmpbox}{\stretchto{%
			\scaleto{%
				\scalerel*[\widthof{\ensuremath{#1}}]{\kern-.6pt\bigwedge\kern-.6pt}%
				{\rule[-\textheight/2]{1ex}{\textheight}}
			}{\textheight}%
		}{0.5ex}}%
	\stackon[1pt]{#1}{\tmpbox}%
}

\usepackage{xparse}
\makeatletter
\RenewDocumentCommand{\title}{om}{%
\IfNoValueTF{#1}
{\gdef\shorttitle{}}
{\gdef\shorttitle{#1}}%
\gdef\@title{#2}%
}
\makeatother

\makeatletter
\def\@tocline#1#2#3#4#5#6#7{\relax
\ifnum #1>\c@tocdepth 
\else
\par \addpenalty\@secpenalty\addvspace{#2}%
\begingroup \hyphenpenalty\@M
\@ifempty{#4}{%
	\@tempdima\csname r@tocindent\number#1\endcsname\relax
}{%
	\@tempdima#4\relax
}%
\parindent\z@ \leftskip#3\relax \advance\leftskip\@tempdima\relax
\rightskip\@pnumwidth plus4em \parfillskip-\@pnumwidth
#5\leavevmode\hskip-\@tempdima
\ifcase #1
\or\or \hskip 2em \or \hskip 2em \else \hskip 3em \fi%
#6\nobreak\relax
\hfill\hbox to\@pnumwidth{\@tocpagenum{#7}}\par
\nobreak
\endgroup
\fi}
\makeatother

\title[Harmonic analysis and circuit approximation on quantum input spaces]{Harmonic analysis and circuit approximation of\\Sobolev functions on quantum input spaces}

\author[P.~Schmocker]{Philipp Schmocker}
\address{Department of Mathematics, ETH Zurich, Switzerland}
\email{philipp.schmocker@math.ethz.ch}

\thanks{\textit{MSC2020 Subject Classification:} primary 43A75, 43A85, 41A25; secondary 46E35, 68Q12, 81P68}

\date{\today}

\begin{document}

\begin{abstract}
	We derive approximation rates for Sobolev functions defined on the complex unit sphere and complex projective space, representing quantum states, and on the special and projective unitary groups, representing quantum unitary operations. By using harmonic analysis on compact groups and compact homogeneous spaces, we characterize the corresponding Peter--Weyl decompositions and Sobolev spaces via irreducible representations and the Laplace--Beltrami operator. Combining spectral truncation with finite-depth quantum-circuit approximation, invoking the Solovay--Kitaev theorem, we establish quantitative rates that separate the spectral truncation error from the circuit approximation error.
\end{abstract}

\vspace{-1.em}
\maketitle

\vspace{-0.6cm}
{
	\hypersetup{linkcolor=black}
	\tableofcontents
}
\vspace{-0.8cm}

\section{Introduction}
\label{sec:intro}

The geometry of quantum computing is intrinsically non-Euclidean. For $d = 2^n$, $n$-qubit states can be represented by normalized vectors on the complex unit sphere $\mathbb{S}^{2d-1}$, while physical $n$-qubit states, obtained by identifying vectors that differ only by a global phase, form the complex projective space $\mathbb{CP}^{d-1}$. In particular, for a single qubit with $n=1$, the space $\mathbb{CP}^{1}$ is naturally identified with the Bloch sphere $\mathbb{S}^2$. The corresponding quantum unitary operations of $n$-qubit state spaces are represented by the unitary group $\mathcal{U}(d)$ and the special unitary group $\mathcal{SU}(d)$, while unitary operations up to global phase form the projective unitary group $\mathcal{PU}(d)$. Hence, functions of quantum states and quantum unitary operations are naturally defined on these non-linear manifolds rather than on Euclidean domains.

Much of the existing literature, however, considers classical functions on Euclidean domains and embeds them into a quantum-computing framework through approximation by parametrized quantum circuits (PQCs), a central tool in quantum machine learning. Early works established universal approximation properties through data re-uploading and characterized the resulting quantum models in terms of truncated Fourier series \cite{perezsalinas20,perezsalinas21,schuld21}. More recent contributions derive quantitative approximation rates in terms of circuit width, circuit depth, and the number of qubits for a range of function classes, e.g., Korobov functions \cite{aftab24}, Fourier-integrable functions~\cite{gonon25}, Sobolev functions~\cite{manzano25}, smooth and H\"older-regular functions \cite{yu24}, and periodic functions~\cite{neufeld25}.

In contrast, we consider functions $f: G \rightarrow \mathbb{C}$ whose inputs are themselves quantum states or quantum unitary operations and are therefore defined on either the compact homogeneous spaces $G \in \lbrace \mathbb{S}^{2d-1}, \mathbb{CP}^{d-1} \rbrace$ or the compact Lie groups $G \in \lbrace \mathcal{SU}(d), \mathcal{PU}(d) \rbrace$. Combining spectral approximation with finite-depth quantum-circuit synthesis of the input, we obtain, for every Sobolev function $f\in H^s(G)$, spectral cutoff $K$, and circuit depth $L$, an approximation rate of the form
\begin{equation}
	\label{eq:intro}
	\Vert f - \mathcal{Q}_{K,L} f \Vert_{L^2(G)} \lesssim_{m,s} \left( (1+K)^{-s} + (1+K)^{\frac{m}{2}+1} \varepsilon_n(L) \right) \Vert f \Vert_{H^s(G)},
\end{equation}
where $m := \dim(G)$ is the real dimension of $G$ and $\varepsilon_n(L)$ denotes the worst-case circuit approximation error at depth $L$. Since the number of retained modes satisfies $N_K \asymp_m (1+K)^m$, the two terms in \eqref{eq:intro} scale as $N_K^{-s/m}$ and $N_K^{1/2+1/m} \varepsilon_n(L)$. Thus, the error \eqref{eq:intro} separates into a spectral term governed by the Sobolev regularity of $f$ and a circuit term governed by $\varepsilon_n(L)$.

In order to control the spectral approximation in \eqref{eq:intro}, we use harmonic analysis on compact Lie groups and compact homogeneous spaces. Using the Peter--Weyl decomposition and the Laplace--Beltrami operator, we characterize the corresponding Sobolev spaces and construct spectral projections onto finitely many eigenspaces, which attain the optimal $L^2$-approximation rate for $H^s$-Sobolev functions. For $\mathcal{U}(d)$, $\mathcal{SU}(d)$, and $\mathcal{PU}(d)$, the spectral components are indexed by their irreducible unitary representations, whereas for $\mathbb{S}^{2d-1}$ and $\mathbb{CP}^{d-1}$, they are described by spherical harmonics and phase-invariant bihomogeneous harmonics, respectively.

For the second term in \eqref{eq:intro}, we approximate the inputs by finite-depth quantum circuits. Fixing a finite universal gate set, the Solovay--Kitaev theorem (see \cite{kitaev97} and also \cite{dawson06,kuperberg23}) yields finite-depth approximations of arbitrary matrices in $\mathcal{SU}(d)$ and $\mathcal{PU}(d)$ to prescribed accuracy, thereby providing an explicit decay rate for the circuit error $\varepsilon_n(L)$ in terms of the circuit depth~$L$. Applying these circuit-generated unitary operations to a fixed reference state induces corresponding approximations on $\mathbb{S}^{2d-1}$ and $\mathbb{CP}^{d-1}$. The resulting effect on the function values is then uniformly controlled by a Bernstein--Nikolskii inequality for the spectral truncation.

This provides a unified approximation framework for functions defined directly on spaces of quantum states and quantum unitary operations. In contrast to approximation on Euclidean input domains, the geometry and symmetries of these spaces determine both the spectral decomposition and the resulting quantitative rates. The bound~\eqref{eq:intro} also reveals the interplay between spectral resolution and circuit depth: increasing the spectral cutoff $K$ reduces the truncation error but amplifies the circuit contribution, which must be compensated by increasing the circuit depth~$L$. Convergence is therefore obtained by suitably balancing $K \rightarrow \infty$ and $L \rightarrow \infty$. In two numerical examples, we illustrate that the approximation error exhibits the theoretically predicted dependence \eqref{eq:intro} on the spectral resolution and the maximal circuit depth.

\subsection{Outline}

The remainder of this article is structured as follows. In Section~\ref{sec:harm_sob}, we review harmonic analysis on compact groups and compact homogeneous spaces and introduce the corresponding Sobolev spaces. In Section~\ref{sec:quantum}, we recall the basic notions of quantum computing and the geometric description of quantum states and quantum unitary operations. In Section~\ref{sec:rate}, we derive approximation rates for Sobolev functions of quantum states and quantum unitary operations. Finally, in Section~\ref{sec:numerics}, we present two numerical examples illustrating the theoretical error bounds.

\subsection{Notation}
\label{sec:notation}

As usual, $\mathbb{N} := \lbrace 1,2,\ldots \rbrace$ and $\mathbb{N}_0 := \mathbb{N} \cup \lbrace 0 \rbrace$ denote the sets of natural numbers, $\mathbb{Z}$ the set of integers, $\mathbb{R}$ the real line, and $\mathbb{C}$ the set of complex numbers (with $\mathbf{i} := \sqrt{-1}$ and complex conjugate $\overline{a+b\mathbf{i}} := a-b\mathbf{i}$, for all $a,b \in \mathbb{R}$). Moreover, $a \lesssim_{c_1,\ldots,c_d} b$ means $a \leq Cb$ for some $C > 0$ depending only on $c_1,\ldots,c_d$, while $a \asymp_{c_1,\ldots,c_d} b$ means $a \lesssim_{c_1,\ldots,c_d} b$ and $b \lesssim_{c_1,\ldots,c_d} a$.

For any $d \in \mathbb{N}$, we denote by $\mathbb{C}^d$ the complex Euclidean space equipped with the norm $\vert x \vert := \big( \sum_{i=1}^d \vert x_i \vert^2 \big)^{1/2}$ induced by the inner product $\langle x, y \rangle_{\mathbb{C}^d} := y^* x := \sum_{i=1}^d x_i \overline{y_i}$, for $x := (x_1,\ldots,x_d) \in \mathbb{C}^d$ and $y := (y_1,\ldots,y_d) \in \mathbb{C}^d$. For any $m,n \in \mathbb{N}$, we denote by $\mathbb{C}^{m \times n}$ the vector space of complex ($m \times n$)-matrices equipped with the Hilbert--Schmidt (or Frobenius) norm $\Vert A \Vert_{\HS} := \big( \sum_{i=1}^m \sum_{j=1}^n \vert a_{i,j} \vert^2 \big)^{1/2}$ induced by the inner product $\langle A,B \rangle_{\mathbb{C}^{m \times n}} := \trace(A B^*) := \sum_{i=1}^m \sum_{j=1}^n a_{i,j} \overline{b_{i,j}}$, for $A := (a_{i,j})_{i=1,\ldots,m, \, j=1,\ldots,n} \in \mathbb{C}^{m \times n}$ and $B := (b_{i,j})_{i=1,\ldots,m, \, j=1,\ldots,n} \in \mathbb{C}^{m \times n}$. If $A^\top$ denotes the transpose of $A$, then $A^* := \overline{A}^\top$ is the Hermitian adjoint of $A$. Moreover, if $m=n$, we denote by $I_n \in \mathbb{C}^{n \times n}$ the identity matrix, by $\trace(A) := \sum_{i=1}^m a_{i,i}$ the trace of $A \in \mathbb{C}^{m \times m}$, by $\det(A)$ its determinant, and by $A^{\otimes k}$ its $k$-fold tensor product.

Throughout the paper, we only consider complex Hilbert spaces $\mathcal{H} := (\mathcal{H},\langle \cdot,\cdot\rangle_\mathcal{H})$, i.e., complex vector spaces $\mathcal{H}$ that are complete with respect to $\Vert x \Vert_\mathcal{H} := \sqrt{\langle x, x \rangle_\mathcal{H}}$. A family $(e_i)_{i \in I} \subseteq \mathcal{H}$ is called an \emph{orthonormal basis (ONB)} of $\mathcal{H}$ if $\langle e_i,e_j \rangle_\mathcal{H} = \delta_{i,j}$ for all $i,j \in I$ and $\linspan\lbrace e_i: i \in I \rbrace$ is dense in $\mathcal{H}$. In addition, $L(\mathcal{H};\widetilde{\mathcal{H}})$ denotes the Banach space of bounded linear operators $T: \mathcal{H} \rightarrow \widetilde{\mathcal{H}}$ into another Hilbert space $\widetilde{\mathcal{H}}$, equipped with the norm $\Vert T \Vert_{L(\mathcal{H};\widetilde{\mathcal{H}})} := \sup_{x \in \mathcal{H}, \, \Vert x \Vert_{\mathcal{H}} \leq 1} \Vert Tx \Vert_{\widetilde{\mathcal{H}}}$, where $L(\mathcal{H}) := L(\mathcal{H};\mathcal{H})$. In particular, the \emph{dual space} $\mathcal{H}^* := L(\mathcal{H};\mathbb{C})$ of $\mathcal{H}$ consists of bounded linear functionals and satisfies $\mathcal{H}^* \cong \mathcal{H}$ by Riesz representation. In addition, the \emph{adjoint} of $T \in L(\mathcal{H};\widetilde{\mathcal{H}})$ is the unique operator $T^* \in L(\widetilde{\mathcal{H}};\mathcal{H})$ such that $\langle Tx, y \rangle_{\widetilde{\mathcal{H}}} = \langle x, T^*y \rangle_\mathcal{H}$, for all $x \in \mathcal{H}$ and $y \in \widetilde{\mathcal{H}}$. 

For a set $G$ and an operator $\cdot: G \times G \rightarrow G$, a tuple $G := (G,\cdot)$ is called a \emph{group} if
\vspace{-0.15cm}
\begin{enumerate}
	\item $g_1 \cdot (g_2 \cdot g_3) = (g_1 \cdot g_2) \cdot g_3$ for all $g_1,g_2,g_3 \in G$,
	\item there exists a unique neutral element $e \in G$ with $e \cdot g = g \cdot e = g$ for all $g \in G$, and
	\item for every $g \in G$ there exists a unique inverse $g^{-1} \in G$ with $g \cdot g^{-1} = g^{-1} \cdot g = e$.
\end{enumerate}
\vspace{-0.15cm}
If $G$ is a group, then $H$ is called a \emph{subgroup} of $G$ (denoted by $H \leqslant G$) if $H \subseteq G$ is a subset and $(H,\cdot)$ is a group itself. The corresponding quotient space $G/H := \lbrace gH: g \in G \rbrace$ consists of left-cosets $gH := \lbrace gh: h \in H \rbrace$. Moreover, the \emph{center} $Z(G) := \lbrace g \in G: gh = hg \text{ for all } h \in G \rbrace$ of $G$ is a subgroup of $G$. In addition, if $G$ acts on a set $X$, we define $\stab_G(x) := \lbrace g \in G: g \cdot x = x \rbrace$ as the \emph{stabilizer} of $G$ at $x \in X$.

A \emph{topological group} $G := (G,\tau_G)$ is a group $G$ equipped with a Hausdorff topology $\tau_G$ such that multiplication $G \times G \ni (g,h) \mapsto g \cdot h \in G$ and inverse map $G \ni g \mapsto g^{-1} \in G$ are continuous. A \emph{Lie group} is a finite-dimensional smooth manifold $G$ with a compatible group structure such that the multiplication and inverse map are smooth. For a topological group $G$, a \emph{homogeneous $G$-space (or simply homogeneous space)} is a topological space $X$ together with a continuous transitive action $G \times X \rightarrow X$ of $G$. For example, if $H \leqslant G$ is a subgroup (and $G$ is compact), then $G/H$ equipped with the quotient topology is a (compact) homogeneous space, where $G$ acts transitively on $G/H$ via $G \times G/H \ni (g_1,g_2 H) \mapsto (g_1 \cdot g_2) H \in G/H$. Furthermore, if $G$ is a (compact) Lie group and $H \leqslant G$ is a closed subgroup, then $G/H$ equipped with the quotient topology is a (compact) homogeneous manifold. For more details, we refer to the textbooks~\cite{helgason00,folland16}.

If $G$ is a locally compact group, there exists by \cite[Theorem~2.10]{folland16} a left Haar measure $\mu_G$ on $G$, i.e., a non-zero Radon measure $\mu_G$ on $G$ with $\mu_G(xE) = \mu_G(E)$, for all $x \in G$ and Borel sets $E \subseteq G$, where $xE := \lbrace xg: g \in E \rbrace$. Moreover, $C^0(G)$ denotes the vector space of continuous functions $f: G \rightarrow \mathbb{C}$, while $L^p(G)$ represents the vector space of (equivalence classes of) measurable functions $f: G \rightarrow \mathbb{C}$ with $\Vert f \Vert_{L^p(G)} := \big( \int_G \vert f(x) \vert^p \, \mu_G(dx) \big)^{1/p} < \infty$, where $1 \leq p < \infty$.

For any fixed $d \in \mathbb{N}$, the \emph{unitary group (of degree $d$)} is defined as
\begin{equation}
	\label{eq:def_U}
	\mathcal{U}(d) := \left\lbrace U \in \mathbb{C}^{d \times d}: U^* U = U U^* = I_d \right\rbrace,
\end{equation}
which is a compact Lie group of real dimension $d^2$. Moreover, the \emph{special unitary group (of degree $d$)} is given as
\vspace{-0.1cm}
\begin{equation}
	\label{eq:def_SU}
	\mathcal{SU}(d) := \left\lbrace U \in \mathcal{U}(d): \det(U) = 1
	\right\rbrace,
\end{equation}
which is a compact Lie subgroup of $\mathcal{U}(d)$ having real dimension $d^2-1$. In addition, the \emph{projective unitary group (of degree $d$)} is defined by
\begin{equation}
	\label{eq:def_PU}
	\mathcal{PU}(d) := \quotientone{\mathcal{U}(d)}{\mathcal{U}(1)} \cong \quotientone{\mathcal{SU}(d)}{Z(\mathcal{SU}(d))},
\end{equation}
which is a compact Lie group and a compact homogeneous $\mathcal{U}(d)$-space. For $U \in \mathcal{U}(d)$, its equivalence class $[U] \in \mathcal{PU}(d)$ is the left coset $[U] := U \, \mathcal{U}(1) = \lbrace e^{\mathbf{i} \theta} U: \theta \in \mathbb{R} \rbrace$, where $\mathcal{U}(1)$ is in \eqref{eq:def_PU} identified with the central subgroup $\lbrace e^{\mathbf{i} \theta} I_d: \theta \in \mathbb{R} \rbrace \subseteq \mathcal{U}(d)$. The isomorphism in \eqref{eq:def_PU} follows from the fact that the canonical homomorphism $\mathcal{SU}(d) \ni U \mapsto [U] \in \mathcal{PU}(d)$ is surjective with kernel $Z(\mathcal{SU}(d)) = \lbrace e^{2\pi \mathbf{i} k/d} I_d: k = 0,\ldots,d-1 \rbrace$.

For any fixed $d \in \mathbb{N} \cap [2,\infty)$, the \emph{unit sphere (in $\mathbb{C}^d$)} is defined by
\begin{equation}
	\label{eq:def_S}
	\mathbb{S}^{2d-1} := \left\lbrace z \in \mathbb{C}^d: \vert z \vert = 1 \right\rbrace \cong \quotientone{\mathcal{U}(d)}{\mathcal{U}(d-1)},
\end{equation}
which is a compact homogeneous $\mathcal{U}(d)$-space, where $\mathcal{U}(d-1)$ is identified with the closed subgroup
\begin{equation}
	\left\lbrace
	\begin{pmatrix}
		1 & 0 \\
		0 & V
	\end{pmatrix}
	\in \mathbb{C}^{d \times d} : V \in \mathcal{U}(d-1) \right\rbrace \leqslant \mathcal{U}(d).
\end{equation}
Indeed, $\mathcal{U}(d)$ acts transitively on $\mathbb{S}^{2d-1}$ via $\mathcal{U}(d) \times \mathbb{S}^{2d-1} \ni (U,z) \mapsto Uz \in \mathbb{S}^{2d-1}$ and the stabilizer of the reference vector $e_1 = (1,0,\ldots,0)^\top \in \mathbb{S}^{2d-1}$ is equal to $\mathcal{U}(d-1)$. Hence, $\mathcal{U}(d)/\mathcal{U}(d-1) \ni U \, \mathcal{U}(d-1) \mapsto Ue_1 \in \mathbb{S}^{2d-1}$ is the $\mathcal{U}(d)$-equivariant diffeomorphism in \eqref{eq:def_S}. Furthermore, the \emph{complex projective space (in $\mathbb{C}^d$)} is defined by
\begin{equation}
	\label{eq:def_CP}
	\mathbb{CP}^{d-1} := \quotientone{\mathbb{C}^d \setminus \lbrace 0 \rbrace}{\sim} \,\cong\, \quotientone{\mathbb{S}^{2d-1}}{\mathcal{U}(1)} \,\cong\, \quotientone{\mathcal{U}(d)}{\mathcal{U}(1) \times \mathcal{U}(d-1)} \,\cong\, \quotientone{\mathcal{SU}(d)}{S(\mathcal{U}(1)\times\mathcal{U}(d-1))},
\end{equation}
which is a compact homogeneous $\mathcal{U}(d)$-space, where $x \sim y$ if and only if $y = \lambda x$ for some $\lambda \in \mathbb{C} \setminus \lbrace 0 \rbrace$, yielding the equivalence classes $[x] := \lbrace \lambda x: \lambda \in \mathbb{C} \setminus \lbrace 0 \rbrace \rbrace$. In \eqref{eq:def_CP}, $\mathcal{U}(1)$ acts on $\mathbb{S}^{2d-1}$ via $\mathcal{U}(1) \times \mathbb{S}^{2d-1} \ni (e^{\mathbf{i} \theta},x) \mapsto e^{\mathbf{i} \theta} x \in \mathbb{S}^{2d-1}$. Since the orbits are the intersections of $[x] \in \mathbb{CP}^{d-1}$ with $\mathbb{S}^{2d-1}$, the map $\mathbb{S}^{2d-1}/\mathcal{U}(1) \ni x \, \mathcal{U}(1) \mapsto [x] \in \mathbb{CP}^{d-1}$ is a diffeomorphism. Moreover, $\mathcal{U}(d)$ acts transitively on $\mathbb{CP}^{d-1}$ via $\mathcal{U}(d) \times \mathbb{CP}^{d-1} \ni (U,[x]) \mapsto [Ux] \in \mathbb{CP}^{d-1}$. Since the stabilizer of $[e_1]$ is
\begin{equation}
	\stab_{\mathcal{U}(d)}([e_1]) = \mathcal{U}(1) \times \mathcal{U}(d-1) \cong \left\lbrace
	\begin{pmatrix}
		z & 0 \\
		0 & V
	\end{pmatrix}
	\in \mathbb{C}^{d \times d}
	:
	z \in \mathcal{U}(1),\,
	V \in \mathcal{U}(d-1)
	\right\rbrace,
\end{equation}
the map $\mathcal{U}(d)/(\mathcal{U}(1) \times \mathcal{U}(d-1)) \ni U \, (\mathcal{U}(1) \times \mathcal{U}(d-1)) \mapsto [Ue_1] \in \mathbb{CP}^{d-1}$ is an $\mathcal{U}(d)$-equivariant diffeomorphism. Similarly, $\mathcal{SU}(d)$ acts transitively on $\mathbb{CP}^{d-1}$ via $\mathcal{SU}(d) \times \mathbb{CP}^{d-1} \ni (U,[x]) \mapsto [Ux] \in \mathbb{CP}^{d-1}$. Since the stabilizer of $[e_1]$ is
\begin{equation}
	\stab_{\mathcal{SU}(d)}([e_1]) = S(\mathcal{U}(1) \times \mathcal{U}(d-1)) :=
	\left\lbrace
	\begin{pmatrix}
		z & 0 \\
		0 & V
	\end{pmatrix}
	:
	z \in \mathcal{U}(1),\,
	V \in \mathcal{U}(d-1),\,
	z\det(V)=1
	\right\rbrace,
\end{equation}
with $S(H)$ denoting the subgroup of all $U \in H \leqslant \mathcal{U}(d)$ with $\det(U) = 1$, the map $\mathcal{SU}(d)/S(\mathcal{U}(1)\times\mathcal{U}(d-1)) \ni U S(\mathcal{U}(1)\times\mathcal{U}(d-1)) \mapsto [Ue_1] \in \mathbb{CP}^{d-1}$ is an $\mathcal{SU}(d)$-equivariant diffeomorphism.

\section{Harmonic analysis and Sobolev spaces on compact groups}
\label{sec:harm_sob}

In this section, we recall the Peter--Weyl decomposition of the space $L^2(G)$ of square-integrable measurable functions $f: G \rightarrow \mathbb{C}$ on a compact group $G$ and introduce the Sobolev space $H^s(G)$.

\subsection{Peter--Weyl decomposition}
\label{sec:pw}

We follow \cite{folland16} and fix a compact group $G$. By \cite[Theorem~2.10]{folland16}, there exists a left Haar measure $\mu_G$ on $G$, i.e., a non-zero Radon measure $\mu_G$ on $G$ with $\mu_G(G) = 1$ and $\mu_G(xE) = \mu_G(E)$, for all $x \in G$ and Borel sets $E \subseteq G$, where $xE := \lbrace xg: g \in E \rbrace$.

A \emph{unitary representation} of $G$ is a homomorphism $\pi: G \rightarrow \mathcal{U}(\mathcal{H}_\pi)$ that is continuous with respect to the strong operator topology, where $\mathcal{H}_\pi$ is a complex Hilbert space and $\mathcal{U}(\mathcal{H}_\pi)$ denotes the group of unitary operators on it. Such a representation $\pi$ is called \emph{irreducible} if $\lbrace 0 \rbrace$ and $\mathcal{H}_\pi$ are the only closed subspaces of $\mathcal{H}_\pi$ that are invariant under $\pi(x)$, for all $x \in G$. Two irreducible unitary representations $\pi: G \rightarrow \mathcal{U}(\mathcal{H}_\pi)$ and $\pi': G \rightarrow \mathcal{U}(\mathcal{H}_{\pi'})$ are equivalent if there exists a unitary operator $U: \mathcal{H}_\pi \rightarrow \mathcal{H}_{\pi'}$ with $U \pi(x) = \pi'(x) U$, for all $x \in G$. The \emph{unitary dual} $\widehat{G}$ of $G$ is the set of equivalence classes $[\pi]$ of irreducible unitary representations of $G$. Since $G$ is compact, every irreducible unitary representation is finite-dimensional with $d_\pi := \dim(\mathcal{H}_\pi) < \infty$.

For an irreducible representation $[\pi] \in \widehat{G}$, the functions $G \ni x \mapsto \phi_{u,v}(x) := \langle \pi(x) u, v \rangle_{\mathcal{H}_\pi} \in \mathbb{C}$, $u,v \in \mathcal{H}_\pi$, are called the \emph{matrix elements} of $\pi$. If $(e^\pi_i)_{i=1,\ldots,d_\pi}$ is an ONB of $\mathcal{H}_\pi$, we define
\begin{equation}
	\label{eq:def_matrix_el}
	\pi_{i,j}(x) := \phi_{e^\pi_j,e^\pi_i}(x) := \langle \pi(x) e^\pi_j, e^\pi_i \rangle_{\mathcal{H}_\pi}, \qquad i,j = 1,\ldots,d_\pi, \, x \in G.
\end{equation}
By defining $\mathcal{E}_\pi := \linspan_\mathbb{C}\lbrace \pi_{i,j}: i,j = 1,\ldots,d_\pi \rbrace$, the algebra $\mathcal{E} := \linspan_\mathbb{C}\big\lbrace \bigcup_{[\pi] \in \widehat{G}} \mathcal{E}_\pi \big\rbrace$ is a dense subset of $C^0(G)$ and $L^p(G)$, for all $1 \leq p < \infty$ (see \cite[Theorem~5.11]{folland16}). 

\begin{theorem}[Peter--Weyl]
	\label{thm:peter_weyl}
	\hspace{-0.1cm}\cite[Theorem~5.12]{folland16}
	Let $G$ be a compact group. Then,
	\begin{equation}
		\label{eq:thm:peter_weyl}
		L^2(G) = \bigoplus_{[\pi] \in \widehat{G}} \mathcal{E}_\pi
	\end{equation}
	and $(\sqrt{d_\pi} \, \pi_{i,j})_{[\pi] \in \widehat{G}, \, i,j = 1,\ldots,d_\pi}$ is an ONB of $L^2(G)$. In particular, for every $f \in L^2(G)$, we have
	\begin{equation}
		f = \sum_{[\pi] \in \widehat{G}} d_\pi \trace\big( \widehat{f}(\pi) \pi(\cdot) \big),
	\end{equation}
	where the sum converges in $L^2(G)$ and $\widehat{f}(\pi) := \int_G f(x) \pi(x)^* \, \mu_G(dx) \in L(\mathcal{H}_\pi) \cong \mathbb{C}^{d_\pi \times d_\pi}$.
\end{theorem}

\begin{example}[Peter--Weyl decomposition of $\mathcal{U}(d)$]
	\label{ex:pw_U}
	For $d \in \mathbb{N}$, the unitary group $\mathcal{U}(d)$ consisting of unitary matrices $U \in \mathbb{C}^{d \times d}$ is a compact Lie group of real dimension $d^2$. Its irreducible unitary representations are finite-dimensional and parametrized by the dominant integral weights
	\begin{equation}
		\label{eq:dom_int_w}
		\Omega_{\mathcal{U}(d)} := \left\lbrace \omega := (\omega_1,\ldots,\omega_d) \in \mathbb{Z}^d: \omega_1 \geq \cdots \geq \omega_d \right\rbrace,
	\end{equation}
	whence $\widehat{\mathcal{U}(d)} \cong \Omega_{\mathcal{U}(d)}$ (see, e.g., \cite[Proposition~15.47]{fulton04}). More precisely, for every $\omega \in \Omega_{\mathcal{U}(d)}$, define $\mu_\omega := (\omega_1-\omega_d,\ldots,\omega_{d-1}-\omega_d,0)$. Then, the irreducible representation $\pi_\omega$ can be realized on the Schur module $\mathcal{H}_\omega := \mathcal{S}_{\mu_\omega}(\mathbb{C}^d)$, which is constructed from tensor powers of $\mathbb{C}^d$ with Young symmetrizers, i.e., $\mathcal{U}(d) \ni U \mapsto \pi_\omega(U) := \det(U)^{\omega_d} \mathcal{S}_{\mu_\omega}(U) \in \mathcal{U}(\mathcal{H}_\omega)$ (see \cite[Section~6.1]{fulton04}), where $d_\omega := \dim(\mathcal{H}_\omega) = \prod_{1 \leq i < j \leq d} \frac{\omega_i-\omega_j+j-i}{j-i}$ by the Weyl dimension formula (see, e.g., \cite[Corollary~24.4]{fulton04}). Choosing an ONB $(e^\omega_i)_{i=1,\ldots,d_\omega}$ of $\mathcal{H}_\omega$, we can define the matrix elements $\pi_{\omega,i,j}(U)$ as above. Hence, by using the identification $\mathcal{H}_\omega \otimes \mathcal{H}_\omega^* \ni u \otimes \ell \mapsto (U \mapsto \ell(\pi_\omega(U)u)) \in \mathcal{E}_{\pi_\omega} := \linspan_\mathbb{C}\lbrace \pi_{\omega,i,j}: i,j = 1,\ldots,d_\omega \rbrace$, we obtain from Theorem~\ref{thm:peter_weyl} the Peter--Weyl decomposition
	\begin{equation}
		\label{eq:pw_Ud}
		L^2(\mathcal{U}(d)) = \bigoplus_{\omega \in \Omega_{\mathcal{U}(d)}} \mathcal{E}_{\pi_\omega} \cong \bigoplus_{\omega \in \Omega_{\mathcal{U}(d)}} \left( \mathcal{H}_\omega \otimes \mathcal{H}_\omega^* \right)
	\end{equation}
	and that $(\sqrt{d_\omega} (\pi_\omega)_{i,j})_{\omega \in \Omega_{\mathcal{U}(d)}, \, i,j=1,\ldots,d_\omega}$ is an ONB of $L^2(\mathcal{U}(d))$.
\end{example}

We now derive a version of Theorem~\ref{thm:peter_weyl} for compact homogeneous spaces. Let $G$ be a compact group and let $H \leqslant G$ be a closed subgroup. We equip the compact homogeneous space $G/H$ with the normalized $G$-invariant measure. For every $[\pi] \in \widehat{G}$, we define $\mathcal{H}_\pi^H := \lbrace u \in \mathcal{H}_\pi: \pi(h)u=u \text{ for all } h \in H \rbrace$ with $d_\pi^H := \dim(\mathcal{H}_\pi^H)$. Moreover, for every $[\pi] \in \widehat{G}_H := \lbrace [\pi] \in \widehat{G}: d_\pi^H \geq 1 \rbrace$, we choose an ONB $(e_i)_{i=1,\ldots,d_\pi}$ of $\mathcal{H}_\pi$ such that $(e_j)_{j=1,\ldots,d_\pi^H}$ is an ONB of $\mathcal{H}_\pi^H$, and set
\vspace{-0.1cm}
\begin{equation} 
	\pi_{i,j}(gH) := \left\langle \pi(g)e_j,e_i \right\rangle_{\mathcal{H}_\pi}, \qquad i=1,\ldots,d_\pi, \quad j=1,\ldots,d_\pi^H. 
	\vspace{-0.07cm}
\end{equation}
This is well defined since $\pi_{i,j}(xhH) = \langle \pi(x) \pi(h) e_j, e_i \rangle_{\mathcal{H}_\pi} = \pi_{i,j}(xH)$, for all $h \in H$. In addition, for every $[\pi] \in \widehat{G}_H$, we define $\mathcal{E}_\pi^H := \linspan\lbrace \pi_{i,j}: i=1,\ldots,d_\pi,\, j=1,\ldots,d_\pi^H \rbrace$.

\begin{corollary}[Peter--Weyl] 
	\label{cor:peter_weyl} 
	Let $G$ be a compact group, let $H \leqslant G$ be a closed subgroup, and equip the compact homogeneous space $G/H$ with the normalized $G$-invariant measure. Then,
	\vspace{-0.1cm}
	\begin{equation} 
		\label{eq:cor:peter_weyl} L^2(G/H) = \bigoplus_{[\pi] \in \widehat{G}_H} \mathcal{E}_\pi^H,
		\vspace{-0.07cm}
	\end{equation}
	and $(\sqrt{d_\pi} \, \pi_{i,j})_{[\pi] \in \widehat{G}_H, \, i=1,\ldots,d_\pi, \, j=1,\ldots,d_\pi^H}$ is an ONB of $L^2(G/H)$. In particular, for $f \in L^2(G/H)$,
	\vspace{-0.1cm}
	\begin{equation} 
		\label{eq:cor:peter_weyl2} 
		f(xH) = \sum_{[\pi] \in \widehat{G}_H} d_\pi \trace\big( \widehat{f}(\pi) \pi(x) \big),
		\vspace{-0.07cm}
	\end{equation}
	where the sum converges in $L^2(G/H)$ and $\trace(\widehat{f}(\pi) \pi(x))$ is independent of the representative $x$ of the coset $xH$, as $\trace(\widehat{f}(\pi) \pi(xh)) = \trace(\pi(h) \widehat{f}(\pi) \pi(x)) = \trace(\widehat{f}(\pi) \pi(x))$ for all $h \in H$.
\end{corollary} 
\begin{proof} 
	By the quotient integral formula (see \cite[Theorem~2.51]{folland16}), the map 
	\vspace{-0.1cm}
	\begin{equation} 
		J:L^2(G/H) \ni f \,\,\, \mapsto \,\,\, Jf := (x \mapsto f(xH)) \in L^2(G)^H := \left\lbrace f \in L^2(G): \,
		\begin{matrix}
			f(\cdot \, h) = f(\cdot) \text{ in } L^2(G) \\
			\text{for all } h \in H
		\end{matrix}
		\right\rbrace
		\vspace{-0.07cm}
	\end{equation}
	is an isometric isomorphism. For every $[\pi] \in \widehat{G}$, the averaging operator $L^2(G) \ni f \mapsto P^H f := \big( x \mapsto \int_H f(xh) \, \mu_H(dh) \big) \in L^2(G)^H$ is the orthogonal projection from $L^2(G)$ onto $L^2(G)^H$. Moreover, $P_\pi^H := \int_H \pi(h) \, \mu_H(dh)$ is the orthogonal projection onto $\mathcal{H}_\pi^H$. Since $P^H(x \mapsto \langle \pi(x) u, v \rangle_{\mathcal{H}_\pi}) = (x \mapsto \langle \pi(x) P_\pi^H u, v \rangle_{\mathcal{H}_\pi})$, the operator $P^H$ preserves every $\mathcal{E}_\pi$ and therefore
	\vspace{-0.1cm}
	\begin{equation}
		\mathcal{E}_\pi \cap L^2(G)^H = \linspan\left\lbrace x \mapsto \left\langle \pi(x) u, v \right\rangle_{\mathcal{H}_\pi}: u \in \mathcal{H}_\pi^H,\, v \in \mathcal{H}_\pi \right\rbrace
		\vspace{-0.07cm}
	\end{equation}
	In particular, $\mathcal{E}_\pi \cap L^2(G)^H$ is trivial whenever $d_\pi^H = 0$, whereas for $d_\pi^H \geq 1$, it is spanned by $(\langle \pi(\cdot) e_j, e_i \rangle_{\mathcal{H}_\pi})_{i=1,\ldots,d_\pi,\, j=1,\ldots,d_\pi^H}$. Combining this with \eqref{eq:thm:peter_weyl} of Theorem~\ref{thm:peter_weyl}, we obtain 
	\vspace{-0.1cm}
	\begin{equation} 
		L^2(G)^H = \bigoplus_{[\pi] \in \widehat{G}_H} \left( \mathcal{E}_\pi \cap L^2(G)^H \right).
		\vspace{-0.07cm}
	\end{equation}
	Transporting this decomposition through $J^{-1}: L^2(G)^H \rightarrow L^2(G/H)$ yields \eqref{eq:cor:peter_weyl}. Moreover, the Schur orthogonality relations and the quotient integral formula imply for every $[\pi],[\pi'] \in \widehat{G}_H$ that
	\vspace{-0.1cm}
	\begin{equation} 
		\int_{G/H} \pi_{i,j}(xH) \overline{\pi'_{k,\ell}(xH)} \, \mu_{G/H}(dxH) = \frac{1}{d_\pi} \delta_{\pi,\pi'} \delta_{i,k} \delta_{j,\ell}.
		\vspace{-0.07cm}
	\end{equation}
	Thus, $(\sqrt{d_\pi} \, \pi_{i,j})_{[\pi] \in \widehat{G}_H, \, i=1,\ldots,d_\pi, \, j=1,\ldots,d_\pi^H}$ is an ONB of $L^2(G/H)$. Finally, expanding every $f \in L^2(G/H)$ with respect to this ONB and using that $\left\langle f,\pi_{i,j} \right\rangle_{L^2(G/H)} = \int_G f(xH) \overline{\pi_{i,j}(xH)} \, \mu_G(dx) = [\widehat{f}(\pi)]_{j,i}$ (by the quotient integral formula), we conclude that
	\vspace{-0.1cm}
	\begin{equation}
		\label{eq:cor:peter_weyl:proof1}
		f(xH) = \sum_{[\pi] \in \widehat{G}_H} d_\pi \sum_{i=1}^{d_\pi} \sum_{j=1}^{d_\pi^H} \left\langle f,\pi_{i,j} \right\rangle_{L^2(G/H)} \pi_{i,j}(xH) = \sum_{[\pi] \in \widehat{G}_H} d_\pi \sum_{i=1}^{d_\pi} \sum_{j=1}^{d_\pi^H} [\widehat{f}(\pi)]_{j,i} \pi_{i,j}(xH)
		\vspace{-0.07cm}
	\end{equation}
	where the sum converges in $L^2(G/H)$. Since $\pi(h) \widehat{f}(\pi) = \widehat{f}(\pi)$ (by right-$H$-invariance), we have $\range(\widehat{f}(\pi)) \subseteq \mathcal{H}_\pi^H$ and therefore $[\widehat{f}(\pi)]_{j,i}=0$ for $j>d_\pi^H$. This together with \eqref{eq:cor:peter_weyl:proof1} proves \eqref{eq:cor:peter_weyl2}.
\end{proof}

\subsection{Sobolev spaces and spectral approximation}
\label{sec:sob}

In this section, we follow \cite{jost11,hebey00} and introduce Sobolev spaces on a compact Riemannian manifold $(G,\mathbf{g}_G)$ without boundary. 

A \emph{Riemannian manifold} $G := (G,\mathbf{g}_G)$ is a finite-dimensional smooth real manifold $G$ equipped with a Riemannian metric $\mathbf{g}_G := (\mathbf{g}_G(x))_{x \in G}$, i.e., a smoothly varying inner product $\mathbf{g}_G(x) : T_x G \times T_x G \rightarrow \mathbb{R}$ on each tangent space $T_x G$, where $\Vert \cdot \Vert_{\mathbf{g}_G(x)}$ denotes the induced norm. The corresponding distance $d_G: G \times G \rightarrow [0,\infty]$ is defined as the infimum of lengths $L(\gamma) := \int_0^1 \Vert \dot{\gamma}(t) \Vert_{\mathbf{g}_G(\gamma(t))} \, dt$ over all piecewise $C^1$-curves $\gamma: [0,1] \rightarrow G$ with $\gamma(0) = x$ and $\gamma(1) = y$, while the Riemannian volume measure is given by $\nu_G(dx) = \sqrt{\det(\mathbf{g}_{G,i,j}(x))} \, dx$, where $\mathbf{g}_{G,i,j}(x) := (\mathbf{g}_G(x)(\partial_i,\partial_j))_{i,j}$ is the matrix representation of the Riemannian metric $\mathbf{g}_G$ at $x \in G$ in terms of local coordinates $\partial_i := \frac{\partial}{\partial x_i}$. The Christoffel symbols of the Levi--Civita connection are in any chart given by
\begin{equation}
	\Gamma_{i,j}^k = \frac{1}{2} \left( \partial_i \mathbf{g}_{G,m,j} + \partial_j \mathbf{g}_{G,m,i} - \partial_m \mathbf{g}_{G,i,j} \right) \mathbf{g}_G^{m,k},
\end{equation}
where we adopt Einstein's summation convention and $(\mathbf{g}_G^{i,j}(x))_{i,j}$ denotes the inverse of the matrix $(\mathbf{g}_{G,i,j}(x))_{i,j}$. For $k \in \mathbb{N}_0$, we denote by $\nabla^k f := \nabla (\nabla^{k-1} f)$ the $k$-th covariant derivative of a smooth function $f \in C^\infty(G)$, with $\nabla^0 f := f$. In particular, the first and second covariant derivatives are in local coordinates given by $(\nabla f)_i = \partial_i f$ and $(\nabla^2 f)_{i,j} = \partial_{i,j} f - \Gamma_{i,j}^k \partial_k f$, respectively, where $\partial_{i,j} f := \frac{\partial^2 f}{\partial x_i \partial x_j}$ denotes the second partial derivative. For more details, we refer to \cite{jost11,hebey00}.

We now assume that $(G,\mathbf{g}_G)$ is a compact Riemannian manifold without boundary and define the normalized Riemannian volume measure $\mu_G(A) := \frac{\nu_G(A)}{\nu_G(G)}$, for all Borel subsets $A \subseteq G$. For $k \in \mathbb{N}_0$, we define the Sobolev space $H^k_\nabla(G)$ as the completion of $C^\infty(G)$ with respect to
\begin{equation}
	\Vert f \Vert_{H^k_\nabla(G)} := \left( \sum_{j=0}^k \int_G \vert \nabla^j f(x) \vert^2 \, \mu_G(dx) \right)^\frac{1}{2},
\end{equation}
which is by \cite[Proposition~2.1]{hebey00} a Hilbert space under the inner product
\begin{equation}
	\langle f,g \rangle_{H^k_\nabla(G)} := \sum_{j=0}^k \int_G \langle \nabla^j f(x), \nabla^j g(x) \rangle \, \mu_G(dx).
\end{equation}
Note that $\vert \nabla^j f(x) \vert^2 = \mathbf{g}_G^{i_1,k_1}(x) \cdots \mathbf{g}_G^{i_j,k_j}(x)
\, (\nabla^j f)_{i_1,\ldots,i_j}(x) \, \overline{(\nabla^j f)_{k_1,\ldots,k_j}(x)}$ in local coordinates. To define Sobolev spaces of order $s \in \mathbb{R}$, let $\Delta_G$ be the Laplace--Beltrami operator defined by
\begin{equation}
	\Delta_G f(x) := \trace_{\mathbf{g}_G(x)}\left( \nabla^2 f(x) \right),
\end{equation}
for $f \in C^\infty(G)$, which, according to \cite[p.~3]{hebey00}, is given in local coordinates by
\begin{equation}
	\Delta_G f(x) = \mathbf{g}_G^{i,j}(x) \left( \partial_{i,j} f(x) - \Gamma_{i,j}^\ell \partial_\ell f(x) \right) = \frac{\partial_m \left( \sqrt{\det(\mathbf{g}_{G,i,j}(x))} \, \mathbf{g}_G^{m,k}(x) \, \partial_k f(x) \right)}{\sqrt{\det(\mathbf{g}_{G,i,j}(x))}}.
\end{equation}
The operator $-\Delta_G$ is non-negative and essentially self-adjoint on $C^\infty(G) \subseteq L^2(G)$ (see, e.g., \cite[Theorem~2.4]{strichartz83}), whence its closure, still denoted by $-\Delta_G$, is a non-negative self-adjoint operator on $L^2(G)$. Since $G$ is compact, $-\Delta_G$ has discrete spectrum (see \cite[Theorem~I.1]{chavel84}), and there exist an ONB $(\varphi_r)_{r \in \mathbb{N}}$ of $L^2(G)$ consisting of smooth eigenfunctions and eigenvalues $(\lambda_r)_{r \in \mathbb{N}} \subseteq [0,\infty)$ with $-\Delta_G \varphi_r = \lambda_r \varphi_r$, for all $r \in \mathbb{N}$. For every $s \in \mathbb{R}$, we define the Sobolev space of order $s$ by
\begin{equation}
	H^s(G) := \left\lbrace f \in \mathcal{D}'(G): \Vert f \Vert_{H^s(G)} < \infty \right\rbrace,
\end{equation}
equipped with the norm 
\begin{equation}
	\label{eq:def_sob_norm}
	\Vert f \Vert_{H^s(G)} := \left( \sum_{r=1}^{\infty} (1+\lambda_r)^s \vert \langle f, \overline{\varphi_r} \rangle \vert^2 \right)^\frac{1}{2},
\end{equation}
where $\mathcal{D}'(G) := C^\infty(G)'$ denotes the space of distributions (i.e., the continuous dual space of $C^\infty(G)$ equipped with its usual Fr\'echet topology) and $\langle \cdot, \cdot \rangle$ denotes the corresponding dual pairing. In addition, for $s \geq 0$, the fractional power $(1-\Delta_G)^{s/2}$ is by \cite{kato60,komatsu66} defined as
\begin{equation}
	(1-\Delta_G)^\frac{s}{2} f = \sum_{r=1}^\infty (1+\lambda_r)^\frac{s}{2} \langle f,\varphi_r\rangle_{L^2(G)} \, \varphi_r,
\end{equation}
for $f \in \dom((1-\Delta_G)^{s/2}) = H^s(G) \subseteq L^2(G)$, where the sum converges in $L^2(G)$. In particular, $\Vert f \Vert_{H^s(G)} = \Vert (1-\Delta_G)^\frac{s}{2} f \Vert_{L^2(G)}$. Furthermore, if $s = k \in \mathbb{N}_0$, then $H^k(G) = H^k_\nabla(G)$ as sets with equivalent norms $\Vert \cdot \Vert_{H^k(G)}$ and $\Vert \cdot \Vert_{H^k_\nabla(G)}$ (see \cite[Theorem~2.7]{hebey00}). For $s < 0$, the space $H^s(G)$ contains $L^2(G)$ and can equivalently be identified via the dual space of $H^{-s}(G)$.

Now, we assume additionally that $G$ is a compact connected Lie group of dimension $n$ and that $\mathbf{g}_G$ is bi-invariant. By \cite[Corollary~8.1.16]{ledonne25}, the normalized Riemannian volume measure $\mu_G$ is a Haar measure on $G$. For $X \in \mathfrak{g}$, the corresponding left-invariant vector field is given by $X^L(x) := (dL_x)_e X \in T_x G$, where $L_x(y) := x \cdot y$ is the left translation (see \cite[p.~119]{ledonne25}). Moreover, let $(X_j)_{j=1,\ldots,n}$ be an ONB of the Lie algebra $\mathfrak{g} := T_e G$ with respect to the inner product $\mathbf{g}_G(e)$. Since $\mathbf{g}_G$ is left-invariant, we have $\mathbf{g}_G(x)(X_i^L(x),X_j^L(x)) = \mathbf{g}_G(e)(X_i,X_j) = \delta_{i,j}$ and $(X_j^L(x))_{j=1,\ldots,n}$ is an ONB of $T_x G$, for all $x \in G$. In addition, since $\mathbf{g}_G$ is bi-invariant, the Levi--Civita connection satisfies $\nabla_{X^L} Y^L = \frac{1}{2} [X^L,Y^L]$ for all $X,Y \in \mathfrak{g}$ (see \cite[Corollary~8.1.13]{ledonne25}), where $[X^L,Y^L] := X^L Y^L - Y^L X^L$ is the Lie bracket. Hence, the first and second covariant derivatives of $f \in C^\infty(G)$ are for every $x \in G$ and $X,Y \in \mathfrak{g}$ given by
\begin{equation}
	\begin{aligned}
		(\nabla f)_x(X^L) & := X^L f(x), \\
		(\nabla^2 f)_x(X^L,Y^L) & := X^L Y^L f(x) - (\nabla_{X^L} {Y^L}) f(x) = X^L Y^L f(x) - \frac{1}{2} [X^L,Y^L] f(x) \\
		& = \frac{1}{2} \left( X^L Y^L + Y^L X^L \right) f(x).
	\end{aligned}
\end{equation}
In particular, by using $\nabla_{X_j^L} X_j^L = 0$, the Laplace--Beltrami operator of $f \in C^\infty(G)$ satisfies
\begin{equation}
	\Delta_G f = \sum_{j=1}^n \left( \left( X_j^L \right)^2 f - (\nabla_{X_j^L} X_j^L) f \right) = \sum_{j=1}^n \left( X_j^L \right)^2 f.
\end{equation}
Furthermore, for every irreducible unitary representation $[\pi] \in \widehat{G}$, let $d_\pi$ denote the dimension of $\pi$ and define the Fourier transform of $f \in L^1(G)$ by
\begin{equation}
	\label{eq:def_fourier}
	\widehat{f}(\pi) := \int_G f(x) \pi(x)^* \, \mu_G(dx) \in L(\mathcal{H}_\pi) \cong \mathbb{C}^{d_\pi \times d_\pi}.
\end{equation}
Moreover, $(\sqrt{d_\pi} \, \pi_{i,j})_{[\pi] \in \widehat{G}, \, i,j=1,\ldots,d_\pi}$ is by Theorem~\ref{thm:peter_weyl} an ONB of $L^2(G)$, whose elements are eigenfunctions of $-\Delta_G$. Indeed, since $-\Delta_G$ is the Casimir operator associated with an $\operatorname{Ad}(G)$-invariant inner product, it is central (see \cite[Section~4.1.2]{fischer20}), whence, for every $[\pi] \in \widehat{G}$ there exists by Schur's lemma some $\lambda_\pi \geq 0$ with $-\Delta_G \pi_{i,j} = \lambda_\pi \pi_{i,j}$, for all $i,j = 1,\ldots,d_\pi$. Thus, by using the ONB $(\varphi_r)_{r \in \mathbb{N}} := (\sqrt{d_\pi} \, \pi_{i,j})_{ [\pi] \in \widehat{G}, \, i,j=1,\ldots,d_\pi}$ in \eqref{eq:def_sob_norm}, it holds for every $f \in H^s(G)$ that
\begin{equation}
	\Vert f \Vert_{H^s(G)} = \left( \sum_{[\pi] \in \widehat{G}} \sum_{i,j=1}^{d_\pi} (1+\lambda_\pi)^s \big\vert \langle f, \sqrt{d_\pi} \, \overline{\pi_{i,j}} \rangle \big\vert^2 \right)^\frac{1}{2} = \left( \sum_{[\pi] \in \widehat{G}} d_\pi (1+\lambda_\pi)^s \Vert \widehat{f}(\pi) \Vert_{\HS}^2 \right)^\frac{1}{2},
\end{equation}
where we have used \eqref{eq:def_fourier}. Hence, the Sobolev space $H^s(G)$ can be equivalently characterized via Fourier transform, where $(1-\Delta_G)^{s/2}$ acts as a Fourier multiplier with symbol $\pi \mapsto (1+\lambda_\pi)^{s/2} I_{d_\pi}$, which also provides the spectral definition of Sobolev spaces for negative orders $s < 0$.

Now, for $K \geq 0$, we define the spectral projection by
\begin{equation}
	L^2(G) \ni f \quad \mapsto \quad P_K f := \sum_{[\pi] \in \widehat{G} \atop \lambda_\pi \leq K^2} d_\pi
	\trace\big( \widehat{f}(\pi) \pi(\cdot) \big) \in L^2(G).
\end{equation}

\begin{proposition}
	\label{prop:fourier_projection}
	For every $s \geq 0$, $f \in H^s(G)$, and $K \geq 0$, we have
	\begin{equation}
		\Vert f - P_K f \Vert_{L^2(G)} \lesssim_s (1+K)^{-s} \Vert f \Vert_{H^s(G)}.
	\end{equation}
\end{proposition}
\begin{proof}
	By Theorem~\ref{thm:peter_weyl} and Plancherel's identity (see, e.g., \cite[Theorem~7.36]{folland16}), it holds that
	\begin{equation}
		\label{eq:fourier_projection:proof1}
		\begin{aligned}
			\Vert f - P_K f \Vert_{L^2(G)}^2 & = \sum_{[\pi] \in \widehat{G} \atop \lambda_\pi > K^2} d_\pi \Vert \widehat{f}(\pi) \Vert_{\HS}^2 \\
			& \lesssim \left( 1+K^2 \right)^{-s} \sum_{[\pi] \in \widehat{G} \atop \lambda_\pi > K^2} d_\pi (1+\lambda_\pi)^s \Vert \widehat{f}(\pi) \Vert_{\HS}^2 \\
			& \lesssim_s (1+K)^{-2s} \Vert f \Vert_{H^s(G)}^2.
		\end{aligned}
		\vspace{-0.5cm}
	\end{equation}
	\qedhere
\end{proof}

\begin{example}[Sobolev spaces on $\mathcal{U}(d)$]
	\label{ex:sob_U}
	For $d \in \mathbb{N}$, we follow \cite[Chapter~12]{faraut08} and equip $\mathcal{U}(d)$ with the bi-invariant Riemannian metric $\mathbf{g}_{\mathcal{U}(d)}(U)(UX,UY) = -\trace(XY)$ induced by the inner product $\langle X,Y \rangle_{\mathfrak{u}(d)} := -\trace(XY)$ on its Lie algebra $\mathfrak{u}(d) := \lbrace A \in \mathbb{C}^{d \times d}: A^* = -A \rbrace$. For $X \in \mathfrak{u}(d)$, the corresponding left-invariant vector field equals $X^L(U) := (dL_U)_I X = UX \in T_U \mathcal{U}(d)$, for $U \in \mathcal{U}(d)$, whose action on $f \in C^\infty(\mathcal{U}(d))$ is given by $X^L f(U) := \frac{d}{dt}\big\vert_{t=0} f(Ue^{tX})$. Hence, by using that the metric is bi-invariant, we can follow the arguments above to conclude that the first and second covariant derivatives of any $f \in C^\infty(\mathcal{U}(d))$ are given by
	\begin{equation}
		\begin{aligned}
			(\nabla f)_U(X^L(U)) & := X^L f(U), \\
			(\nabla^2 f)_U(X^L(U),Y^L(U)) & := X^L Y^L f(U) - \frac{1}{2} [X,Y]^L f(U) = \frac{1}{2} \left( X^L Y^L + Y^L X^L \right) f(U),
		\end{aligned}
	\end{equation}
	for $X,Y \in \mathfrak{u}(d)$. Hence, by taking the trace of $\nabla^2 f$ with respect to an ONB $(X_j)_{j=1,\ldots,d^2}$ of $\mathfrak{u}(d)$ (see, e.g., \cite[p.~275]{faraut08}), the Laplace--Beltrami operator of any $f \in C^\infty(\mathcal{U}(d))$ is given by
	\begin{equation}
		\Delta_{\mathcal{U}(d)} f = \sum_{j=1}^{d^2} \left( X^L_j \right)^2 f.
	\end{equation}
	By Example~\ref{ex:pw_U}, the Peter--Weyl decomposition of $L^2(\mathcal{U}(d))$ is indexed by $\omega \in \Omega_{\mathcal{U}(d)}$. Choosing an ONB $(e^\omega_i)_{i=1,\ldots,d_\omega}$ of $\mathcal{H}_\omega$, the matrix elements $\pi_{\omega,i,j} := \langle \pi_\omega(\cdot) e^\omega_j, e^\omega_i \rangle_{\mathcal{H}_\omega}$ are eigenfunctions of $-\Delta_{\mathcal{U}(d)}$ with eigenvalue $\lambda_\omega := \sum_{i=1}^d \omega_i(\omega_i+d+1-2i)$ (see \cite[Proposition~12.1.2]{faraut08}), where $d_\omega := \dim(\mathcal{H}_\omega) = \prod_{1 \leq i < j \leq d} \frac{\omega_i-\omega_j+j-i}{j-i}$. Thus, for every $s \geq 0$, the Sobolev space is given by 
	\begin{equation}
		H^s(\mathcal{U}(d)) := \dom\left( (1-\Delta_{\mathcal{U}(d)})^\frac{s}{2} \right) = \left\lbrace f \in \mathcal{D}'(\mathcal{U}(d)): \Vert f \Vert_{H^s(\mathcal{U}(d))} < \infty \right\rbrace,
	\end{equation}
	where the norm $\Vert \cdot \Vert_{H^s(\mathcal{U}(d))}$ is given by
	\begin{equation}
		\Vert f \Vert_{H^s(\mathcal{U}(d))} := \Vert (1-\Delta_{\mathcal{U}(d)})^\frac{s}{2} f \Vert_{L^2(\mathcal{U}(d))} = \left( \sum_{\omega \in \Omega_{\mathcal{U}(d)}} d_\omega (1+\lambda_\omega)^s \Vert \widehat{f}_\omega \Vert_{\HS}^2 \right)^\frac{1}{2},
	\end{equation}
	with $\widehat{f}_\omega := \widehat{f}(\pi_\omega)$. Thus, Proposition~\ref{prop:fourier_projection} implies for every $s \geq 0$, $f \in H^s(\mathcal{U}(d))$, and $K \geq 0$ that
	\begin{equation}
		\Vert f - P_K f \Vert_{L^2(\mathcal{U}(d))} \lesssim_s (1+K)^{-s} \Vert f \Vert_{H^s(\mathcal{U}(d))} \asymp_{d,s} N_K^{-\frac{s}{d^2}} \Vert f \Vert_{H^s(\mathcal{U}(d))},
	\end{equation}
	with $N_K := \dim(\range(P_K)) = \sum_{\omega \in \Omega_{\mathcal{U}(d)}, \, \lambda_\omega \leq K^2} d_\omega^2 \asymp_d (1+K)^{d^2}$ by Weyl's law (see, e.g., \cite[Chapter~VI]{chavel84}) and $\dim(\mathcal{U}(d))=d^2$.
\end{example}

We now extend Proposition~\ref{prop:fourier_projection} to compact homogeneous spaces. For a compact connected Lie group $G$ with bi-invariant Riemannian metric and a closed subgroup $H \leqslant G$, we consider the compact homogeneous space $G/H$ equipped with the induced normal homogeneous Riemannian metric. Since the quotient map $\varpi: G \rightarrow G/H$ is a Riemannian submersion whose fibers are the left cosets $xH$, which are compact and totally geodesic submanifolds of $G$, and thus minimal, the Laplace--Beltrami operator $\Delta_{G/H}$ satisfies for every $f\in C^\infty(G/H)$ that
\begin{equation}
	\label{eq:homo_lapl}
	(\Delta_{G/H} f) \circ \varpi = \Delta_G (f \circ \varpi)
\end{equation}
(see, e.g., \cite[Lemma~3.11]{bessa12}). By identifying functions on $G/H$ with their right-$H$-invariant lifts to $G$, the operator $-\Delta_{G/H}$ agrees with the restriction of $-\Delta_G$ to the right-$H$-invariant subspace. For every $s \geq 0$, we define
\begin{equation}
	H^s(G/H) := \dom\left( (1-\Delta_{G/H})^\frac{s}{2} \right) \quad \text{with} \quad \Vert f \Vert_{H^s(G/H)} := \Vert (1-\Delta_{G/H})^\frac{s}{2} f \Vert_{L^2(G/H)}.
\end{equation}
Moreover, for every $[\pi] \in \widehat{G}_H$, let $\lambda_\pi \geq 0$ denote the eigenvalue of $-\Delta_{G/H}$ on the corresponding $\pi$-isotypic component. Then, for every $K \geq 0$, we define
\begin{equation}
	P_K f := \sum_{[\pi] \in \widehat{G}_H \atop \lambda_\pi \leq K^2} d_\pi
	\trace\big( \widehat{f}(\pi) \pi(\cdot) \big).
\end{equation}

\begin{corollary}
	\label{cor:fourier_projection}
	For a compact connected Lie group $G$ with bi-invariant Riemannian metric and a closed subgroup $H \leqslant G$, let the compact homogeneous space $G/H$ be equipped with the induced normal homogeneous Riemannian metric. Then, for every $s \geq 0$, $f \in H^s(G/H)$, and $K \geq 0$,
	\begin{equation}
		\label{eq:cor:fourier_projection}
		\Vert f - P_K f \Vert_{L^2(G/H)} \lesssim_s (1+K)^{-s} \Vert f \Vert_{H^s(G/H)}.
	\end{equation}
\end{corollary}
\begin{proof}
	By Corollary~\ref{cor:peter_weyl}, the Peter--Weyl decomposition of $L^2(G/H)$ is obtained by restricting the Peter--Weyl decomposition of $L^2(G)$ to the right-$H$-invariant subspace $L^2(G)^H \subseteq L^2(G)$ introduced in the proof of Corollary~\ref{cor:peter_weyl}. Moreover, for every $[\pi] \in \widehat{G}_H$ and $f \in C^\infty(G/H)$ with $f \circ \varpi \in \linspan\lbrace \pi_{i,j}: i,j=1,\ldots,d_\pi \rbrace$, it holds that
	\begin{equation}
		(-\Delta_{G/H} f) \circ \varpi = -\Delta_G (f \circ \varpi) = \lambda_\pi (f \circ \varpi),
	\end{equation}
	which implies $-\Delta_{G/H} f = \lambda_\pi f$ by surjectivity of $\varpi: G \rightarrow G/H$. Hence, for every $s \geq 0$ and $f \in H^s(G/H)$, it follows that
	\begin{equation}
		\Vert f \Vert_{H^s(G/H)} = \Vert (1-\Delta_{G/H})^\frac{s}{2} f \Vert_{L^2(G/H)} = \left( \sum_{[\pi] \in \widehat{G}_H} d_\pi (1+\lambda_\pi)^s \Vert\widehat f(\pi)\Vert_{\HS}^2 \right)^\frac{1}{2}.
	\end{equation}
	Finally, the approximation rate~\eqref{eq:cor:fourier_projection} follows from the same arguments as in \eqref{eq:fourier_projection:proof1}.
\end{proof}

\section{Quantum gates and circuits}
\label{sec:quantum}

In this section, we introduce the basic notions of quantum computing. To this end, we follow the textbooks \cite{jacquier24,nielsen10} and the article \cite{chen26}, which contain more background on quantum computing.

\subsection{\texorpdfstring{$1$}{1}-qubit states}

For the Hilbert space $\mathcal{H} := \mathbb{C}^2$, we refer to $v \in \mathbb{C}^2$ as a \emph{state}, denoted by the \emph{ket}-notation $\ket{v}$. The inner product of two states $v,w \in \mathbb{C}^2$ is represented by the \emph{bra-ket}-notation $\braket{v \vert w} := v^* w = \langle w,v \rangle_{\mathbb{C}^2} \in \mathbb{C}$. Elements $u \in \mathcal{H}^*$ of the dual space $\mathcal{H}^*$ are denoted by the \emph{bra}-notation $\bra{u}$. The action of the dual vector $u \in \mathcal{H}^*$ on a vector $v \in \mathcal{H}$ is represented by $\braket{u \vert v}$. The action of a linear operator $A: \mathcal{H} \rightarrow \mathcal{H}$ on $v \in \mathbb{C}^2$ is denoted by $A \ket{v}$. The operator $A: \mathcal{H} \rightarrow \mathcal{H}$ acts on dual vectors $u \in \mathcal{H}^*$ via $(\bra{u} A) \ket{v} := \bra{u} (A \ket{v}) = \braket{u \vert Av} =: \braket{u \vert A \vert v}$, for all $v \in \mathcal{H}$.

A \emph{$1$-qubit state} is a unit vector $\ket{\psi} \in \mathbb{C}^2$ and therefore an element of the complex unit sphere $\mathbb{S}^3$ introduced in \eqref{eq:def_S}. It can be written as
\begin{equation}
	\begin{aligned}
		\ket{\psi} & = \alpha \ket{0} + \beta \ket{1} \\
		& = e^{\mathbf{i}\gamma} \left( \cos(\theta/2) \ket{0} + e^{\mathbf{i}\varphi} \sin(\theta/2) \ket{1} \right),
	\end{aligned}
\end{equation}
for some $\alpha,\beta \in \mathbb{C}$ with $\vert \alpha \vert^2 + \vert \beta \vert^2 = 1$, or equivalently, for suitable $\gamma,\varphi \in [0,2\pi)$ and $\theta \in [0,\pi]$. If a projective measurement contains the rank-one projector $P_\phi := \ket{\phi}\bra{\phi}$ for some $\ket{\phi} \in \mathbb{S}^3$, then $\left\vert \braket{\phi \vert \psi} \right\vert^2$ is by Born's rule (see, e.g., \cite[Chapter~2.5]{marinescu12}) the probability of obtaining the measurement outcome associated with $P_\phi$ when the system is in the state $\ket{\psi}$. Since $\vert \!\braket{\phi \vert e^{\mathbf{i}\gamma} \psi}\! \vert^2 = \vert e^{\mathbf{i}\gamma} \!\braket{\phi \vert \psi}\! \vert^2 = \vert \!\braket{\phi \vert \psi}\! \vert^2$, the factor $e^{\mathbf{i} \gamma}$ of $\ket{\psi} \in \mathbb{S}^3$ has no observable effect. 

Hence, a \emph{physical $1$-qubit state} is an equivalence class
\begin{equation}
	[\psi] := \left\lbrace e^{\mathbf{i} \gamma} \ket{\psi}: \gamma \in \mathbb{R} \right\rbrace \in \mathbb{CP}^1 \cong \quotientone{\mathbb{S}^3}{\mathcal{U}(1)},
\end{equation}
where the complex projective line $\mathbb{CP}^1$ defined in \eqref{eq:def_CP} can be identified with the Bloch sphere $\mathbb{S}^2$ (see \cite[Section~5.2]{bengtsson17}). For every $[\psi] \in \mathbb{CP}^1$, we choose a representative $\ket{\psi} \in \mathbb{S}^3$ of the form
\begin{equation}
	\ket{\psi} = \cos(\theta/2) \ket{0} + e^{\mathbf{i}\varphi} \sin(\theta/2) \ket{1},
\end{equation}
for some $\theta \in [0,\pi]$ and $\varphi \in [0,2\pi)$.

\subsection{\texorpdfstring{$n$}{n}-qubit states}

On the $n$-fold tensor product $\mathcal{H}_n := \mathcal{H}^{\otimes n} \cong \mathbb{C}^{2^n}$ of $\mathcal{H} := \mathbb{C}^2$, a state $\psi \in \mathcal{H}_n$ is denoted by $\ket{\psi}_n$, where the subscript $n$ emphasizes the ($\log_2$)-dimension of $\mathcal{H}_n \cong \mathbb{C}^{2^n}$. Note that $\mathbb{B}_n := \left\lbrace \ket{i}_n: i = (i_1,\ldots,i_n) \in \lbrace 0,1 \rbrace^n \right\rbrace$ is an ONB of $\mathcal{H}_n$, which is called the computational basis in the literature, where $\ket{i}_n := \ket{i_1} \otimes \cdots \otimes \ket{i_n} := \ket{i_1} \cdots \ket{i_n}$ for $i = (i_1,\ldots,i_n) \in \lbrace 0,1 \rbrace^n$.

For $n \in \mathbb{N}$, a \emph{$n$-qubit state} is a unit vector $\ket{\psi}_n \in \mathcal{H}_n$ and therefore an element of the complex unit sphere $\mathbb{S}^{2 \cdot 2^n-1}$ introduced in \eqref{eq:def_S}. It can be uniquely written as
\begin{equation}
	\ket{\psi}_n = \sum_{i \in \lbrace 0,1 \rbrace^n} \alpha_i \ket{i}_n,
\end{equation}
for some $\alpha_i \in \mathbb{C}$ with $\sum_{i \in \lbrace 0,1 \rbrace^n} \vert\alpha_i\vert^2 = 1$. Analogously to the $1$-qubit case, multiplication by a global phase does not affect physical measurement probabilities. Hence, a \emph{physical $n$-qubit state} is an equivalence class
\begin{equation}
	[\psi] := \left\lbrace e^{\mathbf{i} \gamma} \ket{\psi}_n: \gamma \in \mathbb{R} \right\rbrace \in \mathbb{CP}^{2^n-1} \cong \quotientone{\mathbb{S}^{2 \cdot 2^n-1}}{\mathcal{U}(1)}.
\end{equation}
Moreover, an $n$-qubit state $\ket{\psi}_n \in \mathcal{H}_n$ is called \emph{separable} if there exist $1$-qubit states $\ket{\psi_1},\ldots,\ket{\psi_n} \in \mathbb{S}^3$ such that $\ket{\psi}_n = \ket{\psi_1}\otimes\cdots\otimes\ket{\psi_n}$. Otherwise, $\ket{\psi}_n$ is called \emph{entangled}. At the level of physical $n$-qubit states, the separable states are parametrized by $(\mathbb{CP}^1)^n$, which is for $n \geq 2$ a proper subset of $\mathbb{CP}^{2^n-1}$, meaning that there are many more entangled states.

\subsection{Quantum gates and circuits}

In the standard circuit model of quantum computation, quantum states are manipulated by applying a finite sequence of quantum gates. According to the postulates of quantum mechanics, the evolution of a closed quantum system is unitary. Elementary gates typically act on one or two qubits at a time.

\begin{definition}
	Any element of the set of $1$-qubit and $2$-qubit quantum gates
	\begin{equation}
		\mathbb{A} := \left\lbrace X,Y,Z,H,P(\phi),R_x(\theta),R_y(\theta),R_z(\theta): \,
		\begin{matrix}
			\phi \in (0,2\pi) \\
			\theta \in (0,4\pi)
		\end{matrix}
		\right\rbrace \cup \big\lbrace\! \CNOT, \SWAP \!\big\rbrace \subseteq \mathbb{C}^{2 \times 2} \cup \mathbb{C}^{2^2 \times 2^2}
	\end{equation}
	is called an \emph{elementary gate}, where
	\begin{equation}
		X := 
		\begin{bmatrix}
			0 & 1 \\
			1 & 0
		\end{bmatrix}, \qquad
		Y := 
		\begin{bmatrix}
			0 & -\mathbf{i} \\
			\mathbf{i} & 0
		\end{bmatrix}, \qquad
		Z := 
		\begin{bmatrix}
			1 & 0 \\
			0 & -1
		\end{bmatrix}
	\end{equation}
	denote the Pauli gates, where
	\begin{equation}
		H := \frac{1}{\sqrt{2}} 
		\begin{bmatrix}
			1 & 1 \\
			1 & -1
		\end{bmatrix}, \qquad
		P(\phi) := 
		\begin{bmatrix}
			1 & 0 \\
			0 & e^{\mathbf{i}\phi}
		\end{bmatrix}
	\end{equation}
	denote the Hadamard gate and the phase shift gates with $\phi \in (0,2\pi)$, respectively, where
	\begin{equation}
		\begin{aligned}
			R_x(\theta) & := e^{-\mathbf{i} \frac{\theta}{2} X} = 
			\begin{bmatrix}
				\cos(\theta/2) & -\mathbf{i} \sin(\theta/2) \\
				-\mathbf{i} \sin(\theta/2) & \cos(\theta/2)
			\end{bmatrix}, \qquad
			R_y(\theta) & := e^{-\mathbf{i} \frac{\theta}{2} Y} = 
			\begin{bmatrix}
				\cos(\theta/2) & -\sin(\theta/2) \\
				\sin(\theta/2) & \cos(\theta/2)
			\end{bmatrix}
		\end{aligned}
	\end{equation}
	\begin{equation}
		\hspace{-0.5cm}
		R_z(\theta) := e^{-\mathbf{i} \frac{\theta}{2} Z} = 
		\begin{bmatrix}
			e^{-\mathbf{i} \frac{\theta}{2}} & 0 \\
			0 & e^{\mathbf{i} \frac{\theta}{2}}
		\end{bmatrix},
	\end{equation}
	denote the three families of rotation gates, and where
	\begin{equation}
		\CNOT := 
		\begin{bmatrix}
			1 & 0 & 0 & 0 \\
			0 & 1 & 0 & 0 \\
			0 & 0 & 0 & 1 \\
			0 & 0 & 1 & 0
		\end{bmatrix}, \qquad
		\SWAP :=
		\begin{bmatrix}
			1 & 0 & 0 & 0 \\
			0 & 0 & 1 & 0 \\
			0 & 1 & 0 & 0 \\
			0 & 0 & 0 & 1
		\end{bmatrix}
	\end{equation}
	denote the Controlled-NOT gate and the swap gate, respectively.
\end{definition}

\begin{remark}
	Since every $A \in \mathbb{A}$ is unitary, we define $[\mathbb{A}] = \lbrace [A]: A \in \mathbb{A} \rbrace \subseteq \mathcal{PU}(2) \cup \mathcal{PU}(4)$, where $[A] := A \, \mathcal{U}(1) = \lbrace e^{\mathbf{i} \theta} A: \theta \in \mathbb{R} \rbrace$ (see also Section~\ref{sec:notation}). Moreover, every $A \in \mathbb{A}$ is implemented in practical quantum computing software packages, e.g., IBM's Qiskit \cite{treinish23} and Google's Cirq \cite{circ22}.
\end{remark}

A quantum circuit is obtained by composing and tensoring such gates and identity operators, where the identity operators $I_2 \in \mathbb{C}^{2 \times 2}$ represent the wires connecting the gates.

\begin{definition}
	\label{def:quantum_circuit}
	For $n \in \mathbb{N}$ and $L \in \mathbb{N}_0$, a \emph{quantum circuit} $Q$ acting on $n$ qubits is a matrix $Q \in \mathbb{C}^{2^n \times 2^n}$ of the form
	\begin{equation}
		\label{eq:def:quantum_circuit}
		Q = \prod_{\ell=1}^L (A_{\ell,1} \otimes A_{\ell,2} \otimes \cdots \otimes A_{\ell,n_\ell}) \in \mathbb{C}^{2^n \times 2^n},
	\end{equation}
	with the empty product ($L = 0$) understood as $I_{2^n}$, where $(A_{\ell,m})_{\ell=1,\ldots,L, \, m=1,\ldots,n_\ell} \subseteq \mathbb{A} \cup \lbrace I_2 \rbrace$ satisfy $\prod_{m=1}^{n_\ell} \dim(A_{\ell,m}) = 2^n$ for all $\ell = 1,\ldots,L$. Moreover, $M := \sum_{\ell=1}^L \sum_{m=1}^{n_\ell} \mathds{1}_\mathbb{A}(A_{\ell,m})$ denotes the number of elementary quantum gates used in $Q$, $n$ is the number of qubits used in $Q$, and $L$ represents the depth of the quantum circuit.
\end{definition}

\begin{lemma}
	For $n \in \mathbb{N}$ and $L \in \mathbb{N}_0$, let $Q$ be a quantum circuit of the form \eqref{eq:def:quantum_circuit}. Then, we have $[Q] \in \mathcal{PU}(2^n)$.
\end{lemma}
\begin{proof}
	For every fixed $\ell = 1,\ldots,L$, we define the matrix $A_\ell := A_{\ell,1} \otimes A_{\ell,2} \otimes \cdots \otimes A_{\ell,n_\ell} \in \mathbb{C}^{2^n \times 2^n}$. Then, by using that $A_{\ell,m}$ is unitary, for all $m = 1,\ldots,n_\ell$, we observe that
	\begin{equation}
		A_\ell^* A_\ell = \bigotimes_{m=1}^{n_\ell} \left( A_{\ell,m}^* A_{\ell,m} \right) = \bigotimes_{m=1}^{n_\ell} I_{\dim(A_{\ell,m})} = I_{2^n},
	\end{equation}
	which shows that $A_\ell \in\mathcal{U}(2^n)$. Since $\mathcal{U}(2^n)$ is closed under matrix multiplication, we have $Q \in \mathcal{U}(2^n)$ and therefore $[Q] \in \mathcal{PU}(2^n)$.
\end{proof}

\subsection{Solovay--Kitaev theorem}

A universal gate set allows every element of the underlying unitary group to be approximated by finite products of gates. The Solovay--Kitaev theorem provides a quantitative refinement, showing that the required word length grows only polylogarithmically in the inverse approximation accuracy.

On $\mathcal{SU}(d)$, we consider the geodesic distance given for every $U,V \in \mathcal{SU}(d)$ by
\begin{equation}
	d_{\mathcal{SU}(d)}(U,V) := \min_{X \in \mathfrak{su}(d) \atop e^X = U^*V} \Vert X \Vert_{\HS}.
\end{equation}
Moreover, by identifying $\mathcal{PU}(d) \cong \mathcal{SU}(d)/Z(\mathcal{SU}(d))$, we equip $\mathcal{PU}(d)$ with the quotient metric
\begin{equation}
	d_{\mathcal{PU}(d)}([U],[V]) := \min_{k=0,\ldots,d-1} d_{\mathcal{SU}(d)}\left( U, e^{2\pi\mathbf{i}k/d} V \right),
\end{equation}
for $U,V \in \mathcal{SU}(d)$. In order to treat both cases simultaneously, let $G \in \lbrace \mathcal{SU}(d), \mathcal{PU}(d) \rbrace$. Moreover, for a finite subset $\mathcal{A} \subseteq G$ that is closed under inverses and $L \in \mathbb{N}_0$, we define
\begin{equation}
	\mathcal{A}_L := \left\lbrace A_1 \cdots A_\ell: 0 \leq \ell \leq L, \, A_1,\ldots,A_\ell \in \mathcal{A} \right\rbrace,
\end{equation}
where the empty product ($L = 0$) is understood as the identity element of $G$. In addition, we define $\langle \mathcal{A} \rangle := \bigcup_{L \in \mathbb{N}_0} \mathcal{A}_L$ as the subgroup generated by $\mathcal{A}$.

\begin{theorem}[Solovay--Kitaev]
	\label{thm:sk}
	\hspace{-0.1cm}\cite[Theorem~1]{dawson06} \& \cite[Theorem~1.2]{kuperberg23}.
	Let $d \geq 2$, let $G \in \lbrace \mathcal{SU}(d),\mathcal{PU}(d) \rbrace$, and let $\mathcal{A} \subseteq G$ be a finite subset such that $A^{-1} \in \mathcal{A}$, for all $A \in \mathcal{A}$, and $\langle \mathcal{A} \rangle$ is dense in $G$. Then, there exist constants $C > 0$, $c > 1$, and $\varepsilon_0 > 0$ such that for every $U \in G$ and $\varepsilon \in (0,\varepsilon_0]$ there exist $L \in \mathbb{N}_0$ and $V \in \mathcal{A}_L$ satisfying
	\begin{equation}
		d_G(U,V) \leq \varepsilon \qquad \text{and} \qquad L \leq C \log\left( \varepsilon^{-1} \right)^c.
	\end{equation}
\end{theorem}
\begin{proof}
	For $G=\mathcal{SU}(d)$, the conclusion follows from \cite[Theorem~1]{dawson06} using equivalence of the metrics. For $G = \mathcal{PU}(d)$, note that $\mathcal{PU}(d) \cong
	\mathcal{SU}(d)/Z(\mathcal{SU}(d))$ is a connected compact semisimple Lie group with Lie algebra $\mathfrak{pu}(d) \cong \mathfrak{su}(d)$. Its quotient Riemannian metric is bi-invariant, whence left-invariant. Hence, the conclusion follows from \cite[Theorem~1.2]{kuperberg23}, where the distance-dependent factor occurring there is uniformly bounded as $\mathcal{PU}(d)$ is compact.
\end{proof}

The constants in Theorem~\ref{thm:sk} are not universal. More precisely, the constant $C$ depends on $d$, $\mathcal{A}$ (but not on $U$ and $\varepsilon$), whereas the exponent $c$ is independent of $U$ and $\varepsilon$. For $G = \mathcal{SU}(d)$, the classical construction of \cite{dawson06} gives $c = \frac{\log(5)}{\log(3/2)} \approx 3.96936$. More recently, \cite[Theorems~1.1--1.2]{kuperberg23} improved this to any $c > \log_\varphi(2) \approx 1.44042$ with $\varphi = \frac{1+\sqrt{5}}{2}$, for both $G = \mathcal{SU}(d)$ and $G = \mathcal{PU}(d)$.

We now apply Theorem~\ref{thm:sk} to quantum circuits introduced in Definition~\ref{def:quantum_circuit}. However, the set of elementary gates $\mathbb{A}$ is not finite as the matrices $P(\phi)$ and $R_x(\theta),R_y(\theta),R_z(\theta)$ are parametrized by $\phi \in (0,2\pi)$ and $\theta \in (0,4\pi)$. Therefore, for $n \in \mathbb{N}$, we consider the finite set 
\begin{equation}
	\mathcal{A} := \left\lbrace [A_1 \otimes \cdots \otimes A_J]: A_1,\ldots,A_J \in \widetilde{\mathbb{A}}, \, \prod_{j=1}^J \dim(A_j) = 2^n \right\rbrace,
\end{equation}
where $\widetilde{\mathbb{A}} := \left\lbrace I_2,H,P(\pi/4),P(\pi/4)^{-1},\CNOT,\SWAP \right\rbrace$.

\begin{corollary}[Solovay--Kitaev for quantum circuits]
	\label{cor:sk}
	Let $n \in \mathbb{N}$. There exist constants $C_n > 0$, $c > 1$, and $\varepsilon_{0,n} > 0$ such that, for every $[U] \in \mathcal{PU}(2^n)$ and $\varepsilon \in (0,\varepsilon_{0,n}]$ there exists an $n$-qubit quantum circuit $Q$ of depth $L$ such that
	\begin{equation}
		\label{eq:cor:sk1}
		d_{\mathcal{PU}(2^n)}([U],[Q]) \leq \varepsilon \qquad \text{and} \qquad L \leq C_n \log\left( \varepsilon^{-1} \right)^c.
	\end{equation}
	Moreover, if $M$ denotes the number of elementary gates used in $Q$, then 
	\begin{equation}
		\label{eq:cor:sk2}
		M \leq nL \leq n C_n \log\left( \varepsilon^{-1} \right)^c.
	\end{equation}
\end{corollary}
\begin{proof}
	First, we observe that $\mathcal{A}$ is finite. Moreover, by using $H^{-1} = H$, $\CNOT^{-1} = \CNOT$, and $\SWAP^{-1} = \SWAP$, it follows for every $A = [A_1 \otimes \cdots \otimes A_J] \in \mathcal{A}$ with $A_1,\ldots,A_j \in \widetilde{\mathbb{A}}$ that 
	\begin{equation}
		A^{-1} = [A_1^{-1} \otimes \cdots \otimes A_J^{-1}] \in \mathcal{A},
	\end{equation}
	which shows that $\mathcal{A}$ is closed under inverses. In order to prove that $\langle \mathcal{A} \rangle$ is dense in $\mathcal{PU}(2^n)$, we first observe that $[H]$ and $[P(\pi/4)]$ generate a dense subgroup of $\mathcal{PU}(2)$ (see, e.g., \cite[Section~3]{boykin99}). Since adjacent $\SWAP$ gates allow arbitrary permutations of the qubits, compositions of elements of $\mathcal{A}$ can implement $\CNOT$ on any ordered pair of qubits. Hence, using that arbitrary single-qubit gates together with $\CNOT$ generate all unitary operations on $n$-qubits (see \cite{barenco95}), we obtain $\overline{\langle \mathcal{A} \rangle} = \mathcal{PU}(2^n)$. Thus, \eqref{eq:cor:sk1} follows from Theorem~\ref{thm:sk}, and since every element of $\mathcal{A}$ represents one circuit layer containing at most $n$ elementary gates, we have $M \leq nL$, which yields \eqref{eq:cor:sk2}.
\end{proof}

\section{Approximation rates for Sobolev functions on quantum input spaces}
\label{sec:rate}

In this section, we derive approximation rates for Sobolev functions depending on quantum unitary operations and quantum states. Throughout this section, we fix some $n \in \mathbb{N}$, set $d := 2^n$, and distinguish between functions of special unitary operations in $\mathcal{SU}(d)$ (see Section~\ref{sec:rate_SU}), functions of unitary operations up to global phase in $\mathcal{PU}(d)$ (see Section~\ref{sec:rate_PU}), functions of $n$-qubit states in $\mathbb{S}^{2d-1}$ (see Section~\ref{sec:rate_S}), and functions of physical $n$-qubit states in $\mathbb{CP}^{d-1}$ (see Section~\ref{sec:rate_CP}).

In each case, the approximation procedure consists of two steps. First, we approximate a Sobolev function by truncating its spectral expansion, whose error is controlled by the approximation rates in Section~\ref{sec:sob}. Second, we replace the input unitary or state by a circuit-generated approximation with bounded depth. The resulting error is controlled by combining a Bernstein--Nikolskii inequality with the Solovay--Kitaev theorem.

For the latter, fix $G \in \lbrace \mathcal{SU}(d), \mathcal{PU}(d) \rbrace$ and let $\mathcal{A}_n \subseteq G$ be a finite set that is closed under inverses and whose generated subgroup $\langle \mathcal{A}_n \rangle$ is dense in $G$. For $L \in \mathbb{N}_0$, we define
\begin{equation}
	\mathcal{A}_{n,L} := \left\lbrace A_1 \cdots A_\ell: 0 \leq \ell \leq L, \, A_1,\ldots,A_\ell \in \mathcal{A}_n \right\rbrace.
\end{equation}
Moreover, the worst-case approximation error of $G$ by $\mathcal{A}_{n,L}$ is defined for every $L \in \mathbb{N}_0$ as
\begin{equation}
	\varepsilon_n(L) := \sup_{U \in G} \min_{V \in \mathcal{A}_{n,L}} d_G(U,V).
\end{equation}
By Theorem~\ref{thm:sk}, there exist constants $a_n > 0$ and $c > 1$ such that $\varepsilon_n(L) \leq \exp(-a_n L^{1/c})$, for sufficiently large $L \in \mathbb{N}_0$ (see also Corollary~\ref{cor:sk}).

\subsection{Functions of special-unitary operations}
\label{sec:rate_SU}

A special unitary operation on $n$-qubit states is an element of the special unitary group
\begin{equation}
	\mathcal{SU}(d) := \left\lbrace U \in \mathbb{C}^{d \times d}: U^* U = U U^* = I_d, \, \det(U) = 1 \right\rbrace
\end{equation}
defined in \eqref{eq:def_SU}, which is a compact Lie subgroup of $\mathcal{U}(d)$ of real dimension $d^2-1$. We equip $\mathcal{SU}(d)$ with the bi-invariant Riemannian metric $\mathbf{g}_{\mathcal{SU}(d)}(U)(UX,UY) = -\trace(XY)$ induced by the inner product $\langle X,Y \rangle_{\mathfrak{su}(d)} := -\trace(XY)$ on its Lie algebra $\mathfrak{su}(d) := \lbrace A \in \mathbb{C}^{d \times d}: A^* = -A, \, \trace(A) = 0 \rbrace$.

Next, we describe the Peter--Weyl decomposition of $L^2(\mathcal{SU}(d))$. By Example~\ref{ex:pw_U}, every $\omega \in \Omega_{\mathcal{U}(d)} := \lbrace \omega := (\omega_1,\ldots,\omega_d) \in \mathbb{Z}^d: \omega_1 \geq \cdots \geq \omega_d \rbrace$ corresponds to exactly one irreducible unitary representation $\mathcal{U}(d) \ni U \mapsto \pi_\omega(U) := \det(U)^{\omega_d} \mathcal{S}_{\mu_\omega}(U) \in \mathcal{U}(\mathcal{H}_\omega)$ of $\mathcal{U}(d)$, where $\mu_\omega := (\omega_1-\omega_d,\ldots,\omega_{d-1}-\omega_d,0)$ and $d_\omega := \dim(\mathcal{H}_\omega) := \prod_{1 \leq i < j \leq d} \frac{\omega_i-\omega_j+j-i}{j-i}$. In particular, if $\omega \in \Omega_{\mathcal{U}(d)}$ and $\omega' := \omega + k \, (1,\ldots,1) \in \Omega_{\mathcal{U}(d)}$ for some $k \in \mathbb{Z}$, then $\mu_{\omega'} = \mu_\omega$ and $\pi_{\omega'}(U) = \det(U)^{\omega_d+k} \mathcal{S}_{\mu_\omega}(U) = \det(U)^k \pi_\omega(U)$ for all $U \in \mathcal{U}(d)$, implying $\pi_{\omega'}(U) = \pi_\omega(U)$ for all $U \in \mathcal{SU}(d)$. Hence, by defining
\begin{equation}
	\Omega_{\mathcal{SU}(d)} := \left\lbrace \omega := (\omega_1,\ldots,\omega_d) \in \mathbb{Z}^d: \omega_1 \geq \cdots \geq \omega_d = 0 \right\rbrace \subseteq \Omega_{\mathcal{U}(d)},
\end{equation}
every $\omega \in \Omega_{\mathcal{SU}(d)}$ corresponds to exactly one irreducible representation $\pi_\omega\vert_{\mathcal{SU}(d)}: \mathcal{SU}(d) \rightarrow \mathcal{U}(\mathcal{H}_\omega)$. Choosing an ONB $(e^\omega_i)_{i=1,\ldots,d_\omega}$ of $\mathcal{H}_\omega$, the matrix elements $\pi_{\omega,i,j}(U) := \langle \pi_\omega(U) e^\omega_j, e^\omega_i \rangle_{\mathcal{H}_\omega}$ are eigenfunctions of $-\Delta_{\mathcal{U}(d)}$ with common eigenvalue $\sum_{i=1}^d \omega_i (\omega_i + d + 1 - 2i)$ (see Example~\ref{ex:sob_U}). Since $\mathfrak{u}(d) = \mathfrak{su}(d) \oplus \mathbb{R}Z$ with $Z := \frac{\mathbf{i}}{\sqrt{d}} I_d$, we have $\Delta_{\mathcal{U}(d)} = \Delta_{\mathcal{SU}(d)} + (Z^L)^2$, where $\pi_\omega(e^{\mathbf{i} \theta} I_d) = e^{\mathbf{i} \theta \vert\omega\vert_1} I_{\mathcal{H}_\omega}$ implies $(Z^L)^2 \pi_{\omega,i,j} = -\frac{1}{d} \vert \omega \vert_1^2  \pi_{\omega,i,j}$. Thus, $\pi_{\omega,i,j}\vert_{\mathcal{SU}(d)}$ are eigenfunctions of $-\Delta_{\mathcal{SU}(d)} = -\Delta_{\mathcal{U}(d)} + (Z^L)^2$ with common eigenvalue $\lambda_\omega := \sum_{i=1}^d \omega_i (\omega_i + d + 1 - 2i) - \frac{1}{d} \vert \omega \vert_1^2$, where $\vert \omega \vert_1 := \sum_{i=1}^d \omega_i$. Hence, by Theorem~\ref{thm:peter_weyl}, the family $(\sqrt{d_\omega} \, \pi_{\omega,i,j})_{\omega \in \Omega_{\mathcal{SU}(d)}, \, i,j=1,\ldots,d_\omega}$ is an ONB of $L^2(\mathcal{SU}(d))$ and we obtain the Peter--Weyl decomposition
\begin{equation}
	L^2(\mathcal{SU}(d)) = \bigoplus_{\omega \in \Omega_{\mathcal{SU}(d)}} \mathcal{E}_{\pi_\omega\vert_{\mathcal{SU}(d)}} \cong \bigoplus_{\omega \in \Omega_{\mathcal{SU}(d)}} \left( \mathcal{H}_\omega \otimes \mathcal{H}_\omega^* \right).
\end{equation}
For $s \geq 0$, the Sobolev space $H^s(\mathcal{SU}(d))$ is defined analogously to Example~\ref{ex:sob_U} for $\mathcal{U}(d)$, i.e.,
\begin{equation}
	H^s(\mathcal{SU}(d)) := \dom\left( (1-\Delta_{\mathcal{SU}(d)})^\frac{s}{2} \right) = \left\lbrace f \in \mathcal{D}'(\mathcal{SU}(d)): \Vert f \Vert_{H^s(\mathcal{SU}(d))} < \infty \right\rbrace,
\end{equation}
equipped with the norm
\begin{equation}
	\Vert f \Vert_{H^s(\mathcal{SU}(d))} := \Vert (1-\Delta_{\mathcal{SU}(d)})^\frac{s}{2} f \Vert_{L^2(\mathcal{SU}(d))} = \left( \sum_{\omega \in \Omega_{\mathcal{SU}(d)}} d_\omega (1+\lambda_\omega)^s \Vert \widehat{f}_\omega \Vert_{\HS}^2 \right)^\frac{1}{2},
\end{equation}
where $\widehat{f}_\omega := \widehat{f}(\pi_\omega) := \int_{\mathcal{SU}(d)} f(U) \pi_\omega(U)^* \, \mu_{\mathcal{SU}(d)}(dU) \in L(\mathcal{H}_\omega) \cong \mathbb{C}^{d_\omega \times d_\omega}$. For $K \in \mathbb{N}_0$, we define the corresponding spectral projection by
\begin{equation}
	L^2(\mathcal{SU}(d)) \ni f \quad \mapsto \quad P_K f := \sum_{\omega \in \Omega_{\mathcal{SU}(d)} \atop \lambda_\omega \leq K^2} d_\omega \trace\big( \widehat{f}_\omega \pi_\omega(\cdot) \big) \in L^2(\mathcal{SU}(d)),
\end{equation}
whose rank is given by $N_K := \dim(\range(P_K)) = \sum_{\omega \in \Omega_{\mathcal{SU}(d)}, \, \lambda_\omega \leq K^2} d_\omega^2 \asymp_d (1+K)^{d^2-1}$ thanks to Weyl's law and $\dim(\mathcal{SU}(d)) = d^2-1$.

Now, we approximate any input matrix $U \in \mathcal{SU}(d)$ by a quantum circuit. For $L \in \mathbb{N}_0$, choose
\begin{equation}
	Q_L(U) \in \argmin_{V \in \mathcal{A}_{n,L}} d_{\mathcal{SU}(d)}(U,V),
\end{equation}
where ties are resolved by choosing the first minimizer with respect to a fixed enumeration of the finite set $\mathcal{A}_{n,L}$, ensuring that $Q_L: \mathcal{SU}(d) \rightarrow \mathcal{A}_{n,L} \subseteq \mathcal{SU}(d)$ is measurable. Moreover, for every $\omega \in \Omega_{\mathcal{SU}(d)}$, the irreducible representation $\pi_\omega\vert_{\mathcal{SU}(d)}$ can be realized as a sub-representation of a sufficiently high tensor power, i.e., there exists an isometry $I_\omega: \mathcal{H}_\omega \rightarrow (\mathbb{C}^d)^{\otimes \vert\omega\vert_1}$ with $I_\omega \pi_\omega(U) = U^{\otimes \vert\omega\vert_1} I_\omega$ for all $U \in \mathcal{SU}(d)$, which implies for every $i,j=1,\ldots,d_\omega$ that
\begin{equation}
	\pi_{\omega,i,j}(U) = \langle \pi_\omega(U) e_j^\omega, e_i^\omega \rangle_{\mathcal{H}_\omega} = \braket{I_\omega e_i^\omega \vert U^{\otimes \vert\omega\vert_1} I_\omega e_j^\omega}.
\end{equation}
Hence, for every $K,L \in \mathbb{N}_0$, we define the quantum-circuit approximation
\begin{equation}
	\mathcal{SU}(d) \ni U \quad \mapsto \quad (\mathcal{Q}_{K,L} f)(U) := \sum_{\omega \in \Omega_{\mathcal{SU}(d)} \atop \lambda_\omega \leq K^2} d_\omega \sum_{i,j=1}^{d_\omega} \widehat{f}_{\omega,i,j} \braket{I_\omega e_i^\omega \vert Q_L(U)^{\otimes \vert\omega\vert_1} I_\omega e_j^\omega} \in \mathbb{C}, 
\end{equation}
where $\widehat{f}_{\omega,i,j} := \langle \widehat{f}_\omega e_i^\omega, e_j^\omega \rangle_{\mathcal{H}_\omega}$. For every fixed $U \in \mathcal{SU}(d)$, the approximation $Q_L(U) \in \mathcal{A}_{n,L}$ is a product of at most $L$ elements of the finite generating set $\mathcal{A}_n$. Hence,
$\mathcal{Q}_{K,L} f$ is obtained by evaluating the spectrally truncated expansion of $f$ at the word approximation $Q_L(U)$ of $U$. In particular, the parameter $L$ controls the algebraic complexity of the approximation through the maximal word length in $\mathcal{A}_n$.

\begin{theorem}
	\label{thm:rate_SU}
	Let $s \geq 0$ and $f \in H^s(\mathcal{SU}(d))$. Then, for every $K,L \in \mathbb{N}_0$, we have
	\begin{equation}
		\begin{aligned}
			\Vert f - \mathcal{Q}_{K,L} f \Vert_{L^2(\mathcal{SU}(d))} & \lesssim_{d,s} \left( (1+K)^{-s} + (1+K)^{\frac{d^2+1}{2}} \varepsilon_n(L) \right) \Vert f \Vert_{H^s(\mathcal{SU}(d))} \\
			& \lesssim_{d,s} \left( N_K^{-\frac{s}{d^2-1}} + N_K^\frac{d^2+1}{2(d^2-1)} \varepsilon_n(L) \right) \Vert f \Vert_{H^s(\mathcal{SU}(d))},
		\end{aligned}
	\end{equation}
	where $N_K \asymp_d (1+K)^{d^2-1}$.
\end{theorem}
\begin{proof}
	Fix $K, L \in \mathbb{N}_0$. By Proposition~\ref{prop:fourier_projection}, we have
	\begin{equation}
		\label{eq:thm:rate_SU:proof1}
		\Vert f - P_K f \Vert_{L^2(\mathcal{SU}(d))} \lesssim_s (1+K)^{-s} \Vert f \Vert_{H^s(\mathcal{SU}(d))}.
	\end{equation}
	Moreover, by definition of $Q_L(U)$, it holds for every $U \in \mathcal{SU}(d)$ that
	\begin{equation}
		\label{eq:thm:rate_SU:proof2}
		d_{\mathcal{SU}(d)}(U,Q_L(U)) = \min_{V \in \mathcal{A}_{n,L}} d_{\mathcal{SU}(d)}(U,V) \leq \varepsilon_n(L).
	\end{equation}
	In addition, since $\mathcal{SU}(d)$ is a compact homogeneous manifold of real dimension $d^2-1$, we can apply the Bernstein--Nikolskii inequality (see, e.g., \cite[Theorems~3.2 and~4.1]{pesenson08}) to obtain
	\begin{equation}
		\label{eq:thm:rate_SU:proof3}
		\Vert \nabla P_K f \Vert_{L^\infty(\mathcal{SU}(d))} \lesssim_d (1+K)^{1+\frac{d^2-1}{2}} \Vert P_K f \Vert_{L^2(\mathcal{SU}(d))} \leq (1+K)^\frac{d^2+1}{2} \Vert f \Vert_{L^2(\mathcal{SU}(d))}.
	\end{equation}
	Hence, by using the mean value theorem and \eqref{eq:thm:rate_SU:proof2}--\eqref{eq:thm:rate_SU:proof3}, it follows for every $U \in \mathcal{SU}(d)$ that
	\begin{equation}
		\label{eq:thm:rate_SU:proof4}
		\begin{aligned}
			\vert (P_K f)(U) - (P_K f)(Q_L(U)) \vert & \leq \Vert \nabla P_K f \Vert_{L^\infty(\mathcal{SU}(d))} d_{\mathcal{SU}(d)}(U,Q_L(U)) \\
			& \lesssim_d (1+K)^\frac{d^2+1}{2} \Vert f \Vert_{L^2(\mathcal{SU}(d))} \varepsilon_n(L).
		\end{aligned}
	\end{equation}
	Finally, by combining \eqref{eq:thm:rate_SU:proof1} with \eqref{eq:thm:rate_SU:proof4}, we conclude that
	\begin{equation}
		\begin{aligned}
			& \Vert f - \mathcal{Q}_{K,L} f \Vert_{L^2(\mathcal{SU}(d))} \leq \Vert f - P_K f \Vert_{L^2(\mathcal{SU}(d))} + \Vert P_K f - \mathcal{Q}_{K,L} f \Vert_{L^2(\mathcal{SU}(d))} \\
			& \lesssim_s (1+K)^{-s} \Vert f \Vert_{H^s(\mathcal{SU}(d))} + \sup_{U \in \mathcal{SU}(d)} \vert (P_K f)(U) - (P_K f)(Q_L(U)) \vert \\
			& \lesssim_d (1+K)^{-s} \Vert f \Vert_{H^s(\mathcal{SU}(d))} + (1+K)^\frac{d^2+1}{2} \varepsilon_n(L) \Vert f \Vert_{L^2(\mathcal{SU}(d))} \\
			& \leq \left( (1+K)^{-s} + (1+K)^\frac{d^2+1}{2} \varepsilon_n(L) \right) \Vert f \Vert_{H^s(\mathcal{SU}(d))},
		\end{aligned}
	\end{equation}
	which completes the proof.
\end{proof}

For $n = 1$ and $d := 2^n = 2$, Theorem~\ref{thm:rate_SU} yields for every $s \geq 0$, $f \in H^s(\mathcal{SU}(2))$, and $K,L \in \mathbb{N}_0$,
\begin{equation}
	\begin{aligned}
		\Vert f - \mathcal{Q}_{K,L} f \Vert_{L^2(\mathcal{SU}(2))} & \lesssim_s \left( (1+K)^{-s} + (1+K)^\frac{5}{2} \varepsilon_1(L) \right) \Vert f \Vert_{H^s(\mathcal{SU}(2))} \\
		& \lesssim_s \left( N_K^{-\frac{s}{3}} + N_K^\frac{5}{6} \varepsilon_1(L) \right) \Vert f \Vert_{H^s(\mathcal{SU}(2))},
	\end{aligned}
\end{equation}
where $N_K \asymp (1+K)^3$.

\subsection{Functions of unitary operations up to global phase}
\label{sec:rate_PU}

A unitary operation up to global phase on $n$-qubit states is represented by an element of the projective unitary group
\begin{equation}
	\mathcal{PU}(d) := \quotientone{\mathcal{U}(d)}{\mathcal{U}(1)} \cong \quotientone{\mathcal{SU}(d)}{Z(\mathcal{SU}(d))}
\end{equation}
introduced in \eqref{eq:def_PU}, which is a compact Lie group of real dimension $d^2-1$ with Lie algebra $\mathfrak{pu}(d) \cong \mathfrak{su}(d)$. We equip $\mathcal{PU}(d)$ with the unique bi-invariant Riemannian metric such that $\mathcal{SU}(d) \ni U \mapsto \varpi(U) := [U] \in \mathcal{PU}(d)$ is a local isometry, i.e., $\mathbf{g}_{\mathcal{PU}(d)}([U])((d\varpi)_U(UX),(d\varpi)_U(UY)) := -\trace(XY)$, for $U \in \mathcal{SU}(d)$ and $X,Y \in \mathfrak{su}(d)$. Under the canonical identification $\mathfrak{pu}(d) \cong \mathfrak{su}(d)$, the metric is induced by the inner product $\langle X,Y \rangle_{\mathfrak{pu}(d)} := -\trace(XY)$ on $\mathfrak{pu}(d)$.

Next, we describe the Peter--Weyl decomposition of $L^2(\mathcal{PU}(d))$. By Section~\ref{sec:rate_SU}, every $\omega \in \Omega_{\mathcal{SU}(d)}$ corresponds to the irreducible unitary representation $\pi_\omega\vert_{\mathcal{SU}(d)}: \mathcal{SU}(d) \rightarrow \mathcal{U}(\mathcal{H}_\omega)$ of $\mathcal{SU}(d)$, where $d_\omega := \dim(\mathcal{H}_\omega) := \prod_{1 \leq i < j \leq d} \frac{\omega_i-\omega_j+j-i}{j-i}$. Since every $U := e^{2\pi \mathbf{i} k/d} I_d \in Z(\mathcal{SU}(d))$ satisfies $\pi_\omega(U) = e^{2\pi \mathbf{i} k \vert\omega\vert_1/d} I_{\mathcal{H}_\omega}$, we observe that $\pi_\omega$ is trivial on $Z(\mathcal{SU}(d))$ if and only if $\vert\omega\vert_1 \md d = 0$. Hence, $\pi_\omega$ descends to an irreducible unitary representation of $\mathcal{PU}(d)$ if and only if $\vert\omega\vert_1 \md d = 0$. Thus, by defining
\begin{equation}
	\Omega_{\mathcal{PU}(d)} := \left\lbrace \omega := (\omega_1,\ldots,\omega_d) \in \mathbb{Z}^d: \omega_1 \geq \cdots \geq \omega_d = 0, \, \vert\omega\vert_1 \md d = 0 \right\rbrace,
\end{equation}
every $\omega \in \Omega_{\mathcal{PU}(d)}$ induces the irreducible unitary representation $\mathcal{PU}(d) \ni [U] \mapsto \widetilde{\pi}_\omega([U]) := \pi_\omega(U) \in \mathcal{U}(\mathcal{H}_\omega)$. Choosing an ONB $(e_i^\omega)_{i=1,\ldots,d_\omega}$ of $\mathcal{H}_\omega$, we use that $\varpi: \mathcal{SU}(d) \rightarrow \mathcal{PU}(d)$ is a local isometry to conclude that the matrix elements $\widetilde{\pi}_{\omega,i,j}([U]) := \langle \widetilde{\pi}_\omega([U]) e_j^\omega, e_i^\omega \rangle_{\mathcal{H}_\omega}$ are eigenfunctions of $-\Delta_{\mathcal{PU}(d)}$ with the same eigenvalues as on $\mathcal{SU}(d)$, i.e., $\lambda_\omega := \sum_{i=1}^d \omega_i(\omega_i+d+1-2i) - \frac{1}{d} \vert\omega\vert_1^2$. Hence, by Theorem~\ref{thm:peter_weyl}, the family $(\sqrt{d_\omega} \, \widetilde{\pi}_{\omega,i,j})_{\omega \in \Omega_{\mathcal{PU}(d)}, \, i,j=1,\ldots,d_\omega}$ is an ONB of $L^2(\mathcal{PU}(d))$ and we obtain the Peter--Weyl decomposition
\begin{equation}
	L^2(\mathcal{PU}(d)) = \bigoplus_{\omega \in \Omega_{\mathcal{PU}(d)}} \mathcal{E}_{\widetilde{\pi}_\omega} \cong \bigoplus_{\omega \in \Omega_{\mathcal{PU}(d)}} \left( \mathcal{H}_\omega \otimes \mathcal{H}_\omega^* \right).
\end{equation}
For $s \geq 0$, the Sobolev space $H^s(\mathcal{PU}(d))$ is defined analogously to Example~\ref{ex:sob_U} for $\mathcal{U}(d)$, i.e.,
\begin{equation}
	H^s(\mathcal{PU}(d)) := \dom\left( (1-\Delta_{\mathcal{PU}(d)})^\frac{s}{2} \right) = \left\lbrace
	f \in \mathcal{D}'(\mathcal{PU}(d)): \Vert f\Vert_{H^s(\mathcal{PU}(d))} < \infty \right\rbrace,
\end{equation}
equipped with the norm
\begin{equation}
	\begin{aligned}
		\Vert f \Vert_{H^s(\mathcal{PU}(d))} := \Vert (1-\Delta_{\mathcal{PU}(d)})^\frac{s}{2} f \Vert_{L^2(\mathcal{PU}(d))} = \left( \sum_{\omega \in \Omega_{\mathcal{PU}(d)}} d_\omega (1+\lambda_\omega)^s \Vert \widehat{f}_\omega \Vert_{\HS}^2 \right)^\frac{1}{2},
	\end{aligned}
\end{equation}
where $\widehat{f}_\omega := \int_{\mathcal{PU}(d)} f([U]) \widetilde{\pi}_\omega([U])^* \, \mu_{\mathcal{PU}(d)}(d[U]) \in L(\mathcal{H}_\omega) \cong \mathbb{C}^{d_\omega \times d_\omega}$. For $K\in\mathbb{N}_0$, we define the corresponding spectral projection by
\begin{equation}
	\label{eq:defP_PU}
	L^2(\mathcal{PU}(d)) \ni f \quad \mapsto \quad P_K f := \sum_{\omega \in \Omega_{\mathcal{PU}(d)} \atop \lambda_\omega \leq K^2} d_\omega \trace\big( \widehat{f}_\omega \widetilde{\pi}_\omega(\cdot) \big) \in L^2(\mathcal{PU}(d)),
\end{equation}
whose rank is given by $N_K := \dim(\range(P_K)) = \sum_{\omega \in \Omega_{\mathcal{PU}(d)}, \, \lambda_\omega \leq K^2} d_\omega^2 \asymp_d (1+K)^{d^2-1}$ thanks to Weyl's law and $\dim(\mathcal{PU}(d))=d^2-1$.

Now, we approximate any input operation $[U] \in \mathcal{PU}(d)$ by a quantum circuit. For $L \in \mathbb{N}_0$, choose
\begin{equation}
	Q_L([U]) \in \argmin_{V \in \mathcal{A}_{n,L}} d_{\mathcal{PU}(d)}([U],V),
\end{equation}
where ties are resolved by choosing the first minimizer with respect to a fixed enumeration of the finite set $\mathcal{A}_{n,L}$, ensuring that $Q_L: \mathcal{PU}(d) \rightarrow \mathcal{A}_{n,L} \subseteq \mathcal{PU}(d)$ is measurable. As in Section~\ref{sec:rate_SU}, there exists for every $\omega \in \Omega_{\mathcal{PU}(d)}$ an isometry $I_\omega: \mathcal{H}_\omega \rightarrow (\mathbb{C}^d)^{\otimes \vert\omega\vert_1}$ with $I_\omega \pi_\omega(U) = U^{\otimes \vert\omega\vert_1} I_\omega$, for all $U\in\mathcal{SU}(d)$. Since $\vert\omega\vert_1 \md d = 0$, the map 
\begin{equation}
	\mathcal{PU}(d) \ni [U] \quad \mapsto \quad \rho_\omega([U]) := U^{\otimes \vert\omega\vert_1} \in \mathcal{U}((\mathbb{C}^d)^{\otimes \vert\omega\vert_1})
\end{equation}
is well-defined, where $U \in \mathcal{SU}(d)$ is any representative of $[U]$. Indeed, if $U' = e^{2\pi\mathbf{i}k/d} U$ for some $k=0,\ldots,d-1$, then $(U')^{\otimes \vert\omega\vert_1} = e^{2\pi\mathbf{i}k\vert\omega\vert_1/d} U^{\otimes \vert\omega\vert_1} = U^{\otimes \vert\omega\vert_1}$. Moreover, $I_\omega \widetilde{\pi}_\omega([U]) = \rho_\omega([U]) I_\omega$, which implies for every $i,j=1,\ldots,d_\omega$ that
\begin{equation}
	\widetilde{\pi}_{\omega,i,j}([U]) = \langle \widetilde{\pi}_\omega([U]) e_j^\omega, e_i^\omega \rangle_{\mathcal{H}_\omega} = \braket{I_\omega e_i^\omega \vert \rho_\omega([U]) I_\omega e_j^\omega}.
\end{equation}
Hence, for every $K,L\in\mathbb{N}_0$, we define
\begin{equation}
	\label{eq:defQ_PU}
	\begin{aligned}
		\mathcal{PU}(d) \ni [U] \quad \mapsto \quad (\mathcal{Q}_{K,L}f)([U]) := \sum_{\omega \in \Omega_{\mathcal{PU}(d)} \atop \lambda_\omega\leq K^2} d_\omega \sum_{i,j=1}^{d_\omega} \widehat{f}_{\omega,i,j} \braket{I_\omega e_i^\omega \vert \rho_\omega(Q_L([U])) I_\omega e_j^\omega} \in\mathbb{C},
	\end{aligned}
\end{equation}
where $\widehat{f}_{\omega,i,j} := \langle \widehat{f}_\omega e_i^\omega,e_j^\omega \rangle_{\mathcal{H}_\omega}$. By construction, $Q_L([U])$ is induced by an $n$-qubit quantum circuit of depth at most $L$. Therefore, $\rho_\omega(Q_L([U]))$ is realized by $\vert\omega\vert_1$ parallel copies of a determinant-one representative of this circuit, and $(\mathcal Q_{K,L}f)([U])$ can be evaluated accordingly, provided that the states $I_\omega e_i^\omega$ are available for all relevant $\omega$ and $i$.

\begin{theorem}
	\label{thm:rate_PU}
	Let $s \geq 0$ and $f \in H^s(\mathcal{PU}(d))$. Then, for every
	$K,L \in \mathbb{N}_0$, we have
	\begin{equation}
		\begin{aligned}
			\Vert f - \mathcal{Q}_{K,L} f \Vert_{L^2(\mathcal{PU}(d))} & \lesssim_{d,s} \left( (1+K)^{-s} + (1+K)^{\frac{d^2+1}{2}} \varepsilon_n(L) \right) \Vert f \Vert_{H^s(\mathcal{PU}(d))} \\
			& \lesssim_{d,s} \left( N_K^{-\frac{s}{d^2-1}} + N_K^\frac{d^2+1}{2(d^2-1)} \varepsilon_n(L) \right) \Vert f \Vert_{H^s(\mathcal{PU}(d))},
		\end{aligned}
	\end{equation}
	where $N_K \asymp_d (1+K)^{d^2-1}$.
\end{theorem}
\begin{proof}
	The conclusion follows from the same arguments as in the proof of Theorem~\ref{thm:rate_SU}, where we use that $\mathcal{PU}(d)$ is also a compact manifold of real dimension $d^2-1$.
\end{proof}

For $n = 1$ and $d := 2^n = 2$, Theorem~\ref{thm:rate_PU} yields for every $s \geq 0$, $f \in H^s(\mathcal{PU}(2))$, and $K,L \in \mathbb{N}_0$,
\begin{equation}
	\label{eq:rate_PU2}
	\begin{aligned}
		\Vert f - \mathcal{Q}_{K,L} f \Vert_{L^2(\mathcal{PU}(2))} & \lesssim_s \left( (1+K)^{-s} + (1+K)^\frac{5}{2} \varepsilon_1(L) \right) \Vert f \Vert_{H^s(\mathcal{PU}(2))} \\
		& \lesssim_s \left( N_K^{-\frac{s}{3}} + N_K^\frac{5}{6} \varepsilon_1(L) \right) \Vert f \Vert_{H^s(\mathcal{PU}(2))}
	\end{aligned}
\end{equation}
where $N_K \asymp (1+K)^3$.

\subsection{Functions of qubit states}
\label{sec:rate_S}

An $n$-qubit state is an element of the complex unit sphere
\begin{equation}
	\mathbb{S}^{2d-1} := \left\lbrace x \in \mathbb{C}^d: \vert x \vert = 1 \right\rbrace \,\cong\, \quotientone{\mathcal{U}(d)}{\mathcal{U}(d-1)} \,\cong\, \quotientone{\mathcal{SU}(d)}{\mathcal{SU}(d-1)}
\end{equation}
introduced in \eqref{eq:def_S}, equipped with the standard round Riemannian metric $\mathbf{g}_{\mathbb{S}^{2d-1}}(x)(v,w) := \re(v^* w)$, for $v,w \in T_x \mathbb{S}^{2d-1}$. Hence, the Levi--Civita connection is given by $\nabla_X Y = P_x(D_X Y)$, with $P_x v := v - \re(x^* v) x$. 

Every smooth function $f \in C^\infty(\mathbb{S}^{2d-1})$ has a smooth extension to $\mathbb{C}^d \setminus \lbrace 0 \rbrace$ defined by $\mathbb{C}^d \setminus \lbrace 0 \rbrace \ni x \mapsto \widetilde{f}(x) := f\big( \frac{x}{\vert x \vert}\big) \in \mathbb{C}$. We identify $\mathbb{C}^d$ with $\mathbb{R}^{2d}$ equipped with the real inner product $\langle v,w \rangle_{\mathbb{R}^{2d}} := \re(\langle v,w\rangle_{\mathbb{C}^d})$. Since the functions considered below may be complex-valued, we complexify the tangent and cotangent bundles, extend the Levi--Civita connection and the Riemannian metric complex-linearly and complex-bilinearly, respectively, and define $\nabla^k f := \nabla^k \re(f) + \mathbf{i}\nabla^k \im(f)$, $k\in\mathbb{N}_0$, where pointwise norms of complexified tensors are taken with respect to the associated Hermitian metric. In particular, for a complex-valued function $f$ on $\mathbb{R}^{2d}$, we write $\nabla{\mathbb{R}^{2d}} f := \nabla_{\mathbb{R}^{2d}} \re(f) + \mathbf{i} \nabla_{\mathbb{R}^{2d}} \im(f)$, so that $Df(x)[v] = \langle \nabla_{\mathbb{R}^{2d}} f(x), v \rangle_{\mathbb{R}^{2d}}$ with the inner product on the right understood as its complex-bilinear extension. Hence, by using that $P_x$ is self-adjoint and $P_x X=X$, it holds for every $X \in T_x \mathbb{S}^{2d-1}$ that
\begin{equation}
	\label{eq:S_covar1}
	(\nabla_{\mathbb{S}^{2d-1}} f)_x(X) = (Xf)(x) = D \widetilde{f}(x)[X] = \langle P_x(\nabla_{\mathbb{R}^{2d}} \widetilde{f}(x)), X \rangle_{\mathbb{R}^{2d}} = \langle
	\nabla_{\mathbb{R}^{2d}} \widetilde{f}(x), X \rangle_{\mathbb{R}^{2d}}.
\end{equation}
Moreover, by extending any $X,Y \in T_x \mathbb{S}^{2d-1}$ locally to smooth tangent vector fields and using $(D_XY)(x) - (\nabla_XY)(x) = -\mathbf{g}_{\mathbb{S}^{2d-1}}(x)(X,Y) \, x$ as well as that $(0,\infty) \ni t \mapsto \widetilde{f}(tx) \in \mathbb{C}$ is constant (implying $\langle \nabla_{\mathbb{R}^{2d}}\widetilde{f}(x), x \rangle_{\mathbb{R}^{2d}} = 0$), we obtain
\begin{equation}
	\label{eq:S_covar2}
	\begin{aligned}
		(\nabla_{\mathbb{S}^{2d-1}}^2 f)_x(X,Y) & = \big( X (Y f) \big)(x) - \big( (\nabla_X Y) f \big)(x) \\
		& = \langle \nabla_{\mathbb{R}^{2d}}^2 \widetilde{f}(x) X, Y \rangle_{\mathbb{R}^{2d}} +
		\langle \nabla_{\mathbb{R}^{2d}} \widetilde{f}(x), D_X Y \rangle_{\mathbb{R}^{2d}} - \langle \nabla_{\mathbb{R}^{2d}}\widetilde{f}(x), \nabla_X Y \rangle_{\mathbb{R}^{2d}}
		\\
		& = \langle \nabla_{\mathbb{R}^{2d}}^2 \widetilde{f}(x) X, Y \rangle_{\mathbb{R}^{2d}} - \langle \nabla_{\mathbb{R}^{2d}}\widetilde{f}(x), x \rangle_{\mathbb{R}^{2d}} \mathbf{g}_{\mathbb{S}^{2d-1}}(x)(X,Y) \\
		& = \langle \nabla_{\mathbb{R}^{2d}}^2 \widetilde{f}(x) X, Y \rangle_{\mathbb{R}^{2d}}.
	\end{aligned}
\end{equation}
Hence, by taking the trace over an ONB $(e_j)_{j=1,\ldots,2d-1}$ of $T_x \mathbb{S}^{2d-1}$, which in turn induces an ONB $(e_1,\ldots,e_{2d-1},x)$ of $\mathbb{R}^{2d}$, and using again that $(0,\infty) \ni t \mapsto \widetilde{f}(tx) \in \mathbb{C}$ is constant (implying $\langle \nabla_{\mathbb{R}^{2d}}^2 \widetilde{f}(x) x, x \rangle_{\mathbb{R}^{2d}} = 0$), the Laplace--Beltrami operator is given by
\begin{equation}
	\label{eq:S_lapl}
	\Delta_{\mathbb{S}^{2d-1}} f(x) = \sum_{j=1}^{2d-1} \langle \nabla_{\mathbb{R}^{2d}}^2 \widetilde{f}(x) e_j,e_j \rangle_{\mathbb{R}^{2d}} = \Delta_{\mathbb{R}^{2d}} \widetilde{f}(x) - \langle \nabla_{\mathbb{R}^{2d}}^2\widetilde{f}(x) x, x \rangle_{\mathbb{R}^{2d}} = \Delta_{\mathbb{R}^{2d}} \widetilde{f}(x).
\end{equation}
This can be used to define the Sobolev spaces $H^s(\mathbb{S}^{2d-1})$ with $s \geq 0$ (see Section~\ref{sec:sob}).

Next, we describe the Peter--Weyl decomposition of $L^2(\mathbb{S}^{2d-1})$. By identifying $L^2(\mathbb{S}^{2d-1})$ with the subspace of right-$\mathcal{U}(d-1)$-invariant functions in $L^2(\mathcal{U}(d))$, the decomposition \eqref{eq:pw_Ud} contains precisely those irreducible representations of $\mathcal{U}(d)$ possessing a non-zero $\mathcal{U}(d-1)$-invariant vector, which (by the branching rule for $\mathcal{U}(d) \supseteq \mathcal{U}(d-1)$, see \cite[Chapter~8]{goodman09}) are parametrized by
\begin{equation}
	\Omega_{\mathbb{S}^{2d-1}} := \left\lbrace (p,0,\ldots,0,-q) \in \mathbb{Z}^d: p,q \in \mathbb{N}_0 \right\rbrace \subseteq \Omega_{\mathcal{U}(d)}.
\end{equation}
For every $p,q \in \mathbb{N}_0$, the irreducible representation $[\pi_{(p,0,\ldots,0,-q)}] \in \widehat{\mathcal{U}(d)}$ is given by $\mathcal{U}(d) \ni U \mapsto \pi_{p,q}(U) := (P \mapsto (z \mapsto P(zU))) \in \mathcal{U}(\mathcal{H}_{p,q}(\mathbb{S}^{2d-1}))$, where $z \in \mathbb{S}^{2d-1}$ is regarded as a row vector, $d_{p,q} := \dim(\mathcal{H}_{p,q}(\mathbb{S}^{2d-1})) = \binom{p+d-1}{p}\binom{q+d-1}{q} - \binom{p+d-2}{p-1}\binom{q+d-2}{q-1} = \frac{(p+q+d-1)(p+d-2)!(q+d-2)!}{p!q!(d-1)!(d-2)!}$ (with $\binom{k}{r} := 0$ if $r < 0$), and $\mathcal{H}_{p,q}(\mathbb{S}^{2d-1}) := \mathcal{H}_{p,q}(\mathbb{C}^d)\vert_{\mathbb{S}^{2d-1}}$ is the restriction of complex spherical harmonic polynomials
\begin{equation}
	\mathcal{H}_{p,q}(\mathbb{C}^d) := \left\lbrace P \in \Pol_{p,q}(\mathbb{C}^d): \Delta_{\mathbb{C}^d} P = 0 \right\rbrace,
\end{equation}
where $\Delta_{\mathbb{C}^d} := 4 \sum_{k=1}^d \frac{\partial^2}{\partial z_k \partial \overline{z}_k}$ and
\begin{equation}
	\Pol_{p,q}(\mathbb{C}^d) := \left\lbrace \mathbb{C}^d \ni z \mapsto \sum_{\alpha,\beta \in \mathbb{N}_0^d, \, \vert\alpha\vert = p, \vert\beta\vert = q} c_{\alpha,\beta} z^\alpha \overline{z}^\beta \in \mathbb{C}: c_{\alpha,\beta} \in \mathbb{C} \right\rbrace.
\end{equation}
Moreover, let $(Y_{p,q,k})_{k=1,\ldots,d_{p,q}}$ be an ONB of $\mathcal{H}_{p,q}(\mathbb{S}^{2d-1})$ consisting of eigenfunctions of $-\Delta_{\mathbb{S}^{2d-1}}$ with eigenvalue $\lambda_{p,q} := (p+q)(p+q+2d-2)$ (see, e.g., \cite{koornwinder73}). Then, $(Y_{p,q,k})_{p,q \in \mathbb{N}_0, \, k=1,\ldots,d_{p,q}}$ is an ONB of $L^2(\mathbb{S}^{2d-1})$ and we obtain the Peter--Weyl decomposition
\begin{equation}
	L^2(\mathbb{S}^{2d-1}) = \widehat{\bigoplus_{p,q \in \mathbb{N}_0}} \mathcal{H}_{p,q}(\mathbb{S}^{2d-1}).
\end{equation}
For $s \geq 0$, we follow Section~\ref{sec:sob} and consider the Sobolev space
\begin{equation}
	H^s(\mathbb{S}^{2d-1}) := \dom\left( (1-\Delta_{\mathbb{S}^{2d-1}})^\frac{s}{2} \right) = \left\lbrace f \in \mathcal{D}'(\mathbb{S}^{2d-1}): \Vert f \Vert_{H^s(\mathbb{S}^{2d-1})} < \infty \right\rbrace,
\end{equation}
which is equipped with the norm
\begin{equation}
	\Vert f \Vert_{H^s(\mathbb{S}^{2d-1})} := \Vert (1-\Delta_{\mathbb{S}^{2d-1}})^\frac{s}{2} f \Vert_{L^2(\mathbb{S}^{2d-1})} = \left( \sum_{p,q\in\mathbb{N}_0} (1+\lambda_{p,q})^s \sum_{k=1}^{d_{p,q}} \big\vert \widehat{f}_{p,q,k} \big\vert^2 \right)^\frac{1}{2},
\end{equation}
with $\widehat{f}_{p,q,k} := \langle f, Y_{p,q,k} \rangle_{L^2(\mathbb{S}^{2d-1})}$. For $K \in \mathbb{N}_0$, we define the corresponding spectral projection by
\begin{equation}
	P_K f := \sum_{p,q \in \mathbb{N}_0 \atop p+q \leq K} \sum_{k=1}^{d_{p,q}} \widehat{f}_{p,q,k} Y_{p,q,k},
\end{equation}
whose rank is given by $N_K := \dim(\range(P_K)) = \binom{K+2d-1}{2d-1} + \binom{K+2d-2}{2d-1} \asymp_d (1+K)^{2d-1}$.

Now, we approximate any input state $x \in \mathbb{S}^{2d-1}$ by a circuit-generated state. For every $x \in \mathbb{S}^{2d-1}$, choose some $U_x \in \mathcal{SU}(d)$ with $x = U_x e_1$ such that $\mathbb{S}^{2d-1} \ni x \mapsto U_x \in \mathcal{SU}(d)$ is measurable, where $e_1 = \ket{0}_n$ is a reference state. Moreover, for $L \in \mathbb{N}_0$, let
\begin{equation}
	C_L(x) \in \argmin_{V \in \mathcal{A}_{n,L}} d_{\mathcal{SU}(d)}(U_x,V),
\end{equation}
where ties are resolved by choosing the first minimizer with respect to a fixed enumeration of the finite set $\mathcal{A}_{n,L}$, ensuring that $C_L: \mathbb{S}^{2d-1} \rightarrow \mathcal{A}_{n,L} \subseteq \mathcal{SU}(d)$ is measurable. Then, $\Pi_L(x) := C_L(x) e_1 \in \mathbb{S}^{2d-1}$ is generated from the reference state $e_1 = \ket{0}_n$ by a quantum circuit of depth at most $L$. Moreover, for every $p,q \in \mathbb{N}_0$ and $k=1,\ldots,d_{p,q}$, the spherical harmonic $Y_{p,q,k}$ is the restriction of a bihomogeneous polynomial of bidegree $(p,q)$, whence there exists a matrix $B_{p,q,k} \in \mathbb{C}^{d^{\max(p,q)} \times d^{\max(p,q)}}$ such that for every $x \in \mathbb{S}^{2d-1}$ it holds that
\begin{equation}
	Y_{p,q,k}(x) = \braket{\psi_{p,q}^q(x) \vert B_{p,q,k} \psi_{p,q}^p(x)},
\end{equation}
where $\psi_{p,q}^r(x) := x^{\otimes r} \otimes e_1^{\otimes (\max(p,q)-r)}$, $r \in \lbrace p,q \rbrace$. Thus, for every $K,L \in \mathbb{N}_0$, we define
\begin{equation}
	\begin{aligned}
		\mathbb{S}^{2d-1} \ni x \quad \mapsto \quad (\mathcal{Q}_{K,L} f)(x) = \sum_{p,q \in \mathbb{N}_0 \atop p+q \leq K} \sum_{k=1}^{d_{p,q}} \widehat{f}_{p,q,k} \braket{\psi_{p,q}^q(\Pi_L(x)) \vert B_{p,q,k} \psi_{p,q}^p(\Pi_L(x))} \in \mathbb{C}.
	\end{aligned}
\end{equation}
Note that $(\mathcal{Q}_{K,L} f)(x)$ can be evaluated using quantum circuits under suitable implementation assumptions on the matrices $B_{p,q,k}$. Indeed, $\Pi_L(x)$ is generated from $e_1 = \ket{0}_n$ by a circuit of depth at most $L$, so the states $\psi_{p,q}^r(\Pi_L(x))$, $r \in \lbrace p,q \rbrace$, can be prepared using at most $\max(p,q)$ parallel copies of this circuit. If the matrices $B_{p,q,k}$ admit efficient quantum implementations, e.g., as linear combinations of efficiently implementable unitaries, the matrix elements $\braket{\psi_{p,q}^q(\Pi_L(x)) \vert B_{p,q,k} \psi_{p,q}^p(\Pi_L(x))}$ can then be estimated using Hadamard-test-type routines.

\begin{theorem}
	\label{thm:rate_S}
	Let $s \geq 0$ and $f \in H^s(\mathbb{S}^{2d-1})$. Then, for every $K,L\in\mathbb{N}_0$, we have
	\begin{equation}
		\begin{aligned}
			\Vert f - \mathcal{Q}_{K,L} f \Vert_{L^2(\mathbb{S}^{2d-1})} & \lesssim_{d,s} \left( (1+K)^{-s} + (1+K)^{\frac{2d+1}{2}} \varepsilon_n(L) \right) \Vert f\Vert_{H^s(\mathbb{S}^{2d-1})} \\
			& \lesssim_{d,s} \left( N_K^{-\frac{s}{2d-1}} + N_K^{\frac{1}{2}+\frac{1}{2d-1}} \varepsilon_n(L) \right) \Vert f \Vert_{H^s(\mathbb{S}^{2d-1})},
		\end{aligned}
	\end{equation}
	where $N_K \asymp_d (1+K)^{2d-1}$.
\end{theorem}
\begin{proof}
	Fix $K,L\in\mathbb{N}_0$. Following the arguments of Corollary~\ref{cor:fourier_projection}, we have
	\begin{equation}
		\label{eq:thm:rate_S:proof1}
		\Vert f - P_K f \Vert_{L^2(\mathbb{S}^{2d-1})} \leq \big( 1 + (K+1)(K+2d-1) \big)^{-\frac{s}{2}} \Vert f \Vert_{H^s(\mathbb{S}^{2d-1})}.
	\end{equation}
	Moreover, for every $x \in \mathbb{S}^{2d-1}$, there exists some $U \in \mathcal{SU}(d)$ with $x = Ue_1$. By the definition of $C_L(x)$ and using that the orbit map $\mathcal{SU}(d) \ni U \mapsto Ue_1 \in \mathbb{S}^{2d-1}$ is $1$-Lipschitz, it holds that
	\begin{equation}
		\label{eq:thm:rate_S:proof2}
		\begin{aligned}
			d_{\mathbb{S}^{2d-1}}(x,\Pi_L(x)) = d_{\mathbb{S}^{2d-1}}(U_xe_1,C_L(x)e_1) \leq d_{\mathcal{SU}(d)}(U_x,C_L(x)) = \min_{V \in \mathcal{A}_{n,L}} d_{\mathcal{SU}(d)}(U_x,V) \leq \varepsilon_n(L).
		\end{aligned}
	\end{equation}
	In addition, since $\mathbb{S}^{2d-1}$ is a compact homogeneous manifold of real dimension $2d-1$, we can apply the Bernstein--Nikolskii inequality (see, e.g., \cite[Theorems~3.2 and~4.1]{pesenson08}) to obtain
	\begin{equation}
		\label{eq:thm:rate_S:proof3}
		\Vert \nabla P_K f \Vert_{L^\infty(\mathbb{S}^{2d-1})} \lesssim_d (1+K)^{1+\frac{2d-1}{2}} \Vert P_K f \Vert_{L^2(\mathbb{S}^{2d-1})} \leq (1+K)^\frac{2d+1}{2} \Vert f \Vert_{L^2(\mathbb{S}^{2d-1})}.
	\end{equation}
	Hence, by using the mean value theorem and the inequalities~\eqref{eq:thm:rate_S:proof3} as well as \eqref{eq:thm:rate_S:proof2}, we have
	\begin{equation}
		\label{eq:thm:rate_S:proof4}
		\begin{aligned}
			\vert (P_K f)(x) - (P_K f)(\Pi_L(x)) \vert & \leq \Vert \nabla P_K f \Vert_{L^\infty(\mathbb{S}^{2d-1})} d_{\mathbb{S}^{2d-1}}(x,\Pi_L(x)) \\
			& \lesssim_d (1+K)^\frac{2d+1}{2} \Vert f \Vert_{L^2(\mathbb{S}^{2d-1})} \varepsilon_n(L).
		\end{aligned}
	\end{equation}
	Finally, by combining \eqref{eq:thm:rate_S:proof1} with \eqref{eq:thm:rate_S:proof4}, it follows that
	\begin{equation}
		\begin{aligned}
			& \Vert f - \mathcal{Q}_{K,L} f \Vert_{L^2(\mathbb{S}^{2d-1})} \leq \Vert f - P_K f \Vert_{L^2(\mathbb{S}^{2d-1})} + \Vert P_K f - \mathcal{Q}_{K,L} f \Vert_{L^2(\mathbb{S}^{2d-1})} \\
			& \leq \big( 1 + (K+1)(K+2d-1) \big)^{-\frac{s}{2}} \Vert f \Vert_{H^s(\mathbb{S}^{2d-1})} + \sup_{x \in \mathbb{S}^{2d-1}} \vert (P_K f)(x) - (P_K f)(\Pi_L(x)) \vert \\
			& \lesssim_d \big( 1 + (K+1)(K+2d-1) \big)^{-\frac{s}{2}} \Vert f \Vert_{H^s(\mathbb{S}^{2d-1})} + (1+K)^\frac{2d+1}{2} \varepsilon_n(L) \Vert f \Vert_{L^2(\mathbb{S}^{2d-1})} \\
			& \lesssim_{d,s} \left( (1+K)^{-s} + (1+K)^{\frac{2d+1}{2}} \varepsilon_n(L) \right) \Vert f\Vert_{H^s(\mathbb{S}^{2d-1})},
		\end{aligned}
	\end{equation}
	which completes the proof.
\end{proof}

For $n=1$ with $d=2^n=2$, we observe that $\mathbb{S}^3$ is diffeomorphic to $\mathcal{SU}(2)$ and Theorem~\ref{thm:rate_S} implies 
\begin{equation}
	\begin{aligned}
		\Vert f - \mathcal{Q}_{K,L} f \Vert_{L^2(\mathbb{S}^3)} & \lesssim_s \left( (1+K)^{-s} + (1+K)^\frac{5}{2} \varepsilon_1(L) \right) \Vert f\Vert_{H^s(\mathbb{S}^3)} \\
		& \lesssim_s \left( N_K^{-\frac{s}{3}} + N_K^\frac{5}{6} \varepsilon_1(L) \right) \Vert f \Vert_{H^s(\mathbb{S}^3)},
	\end{aligned}
\end{equation}
where $N_K \asymp (1+K)^3$.

\subsection{Functions of physical qubit states}
\label{sec:rate_CP}

A physical $n$-qubit state belongs to the complex projective space
\begin{equation}
	\label{eq:CP_identification}
	\mathbb{CP}^{d-1} := \quotientone{\mathbb{C}^d \setminus \lbrace 0 \rbrace}{\sim} \cong \quotientone{\mathbb{S}^{2d-1}}{\mathcal{U}(1)} \cong \quotientone{\mathcal{U}(d)}{\mathcal{U}(1)\times\mathcal{U}(d-1)} \cong \quotientone{\mathcal{SU}(d)}{S(\mathcal{U}(1)\times\mathcal{U}(d-1))}
\end{equation}
introduced in \eqref{eq:def_CP}, where $x \sim y$ if and only if $y = \lambda x$ for some $\lambda \in \mathbb{C} \setminus \lbrace 0 \rbrace$. Let $\mathbf{g}_{\mathrm{FS}}$ be the normalized Fubini--Study metric so that the Hopf projection $(\mathbb{S}^{2d-1},\mathbf{g}_{\mathbb{S}^{2d-1}}) \ni x \mapsto \varpi(x) := [x] \in (\mathbb{CP}^{d-1},\mathbf{g}_{\mathrm{FS}})$ is a Riemannian submersion. We equip $\mathbb{CP}^{d-1}$ with the rescaled metric $\mathbf{g}_{\mathbb{CP}^{d-1}} := 4 \mathbf{g}_{\mathrm{FS}}$ so that $\mathbb{CP}^1$ is isometric to the Bloch sphere $\mathbb{S}^2$. More precisely, for every $x \in \mathbb{S}^{2d-1}$, the tangent space $T_x \mathbb{S}^{2d-1} = \mathscr{H}_x \oplus \mathscr{V}_x$ decomposes orthogonally into the horizontal subspace $\mathscr{H}_x := \lbrace v\in\mathbb{C}^d: x^*v = 0 \rbrace$ and the vertical subspace $\mathscr{V}_x := \ker((d\varpi)_x) = \linspan_{\mathbb{R}}\lbrace \mathbf{i} x \rbrace$. The Fubini--Study metric is given by $\mathbf{g}_{\mathrm{FS}}([x])((d\varpi)_x(v),(d\varpi)_x(w)) := \mathbf{g}_{\mathbb{S}^{2d-1}}(x)(Q_x(v),Q_x(w)) = \re(Q_x(v)^*Q_x(w))$, for $v,w \in T_x \mathbb{S}^{2d-1}$, where $Q_x(v) := (I_d-xx^*) v$ is the orthogonal projection onto $\mathscr{H}_x$. Thus, $(d\varpi)_x\vert_{\mathscr{H}_x}$ is an isometry, whence $\varpi: (\mathbb{S}^{2d-1},\mathbf{g}_{\mathbb{S}^{2d-1}}) \rightarrow (\mathbb{CP}^{d-1},\mathbf{g}_{\mathrm{FS}})$ is a Riemannian submersion.

For $f \in C^\infty(\mathbb{CP}^{d-1})$, define $F := f \circ \varpi \in C^\infty(\mathbb{S}^{2d-1})$. Moreover, for $U,V \in T_{[x]} \mathbb{CP}^{d-1}$, let $X,Y \in \mathscr{H}_x$ denote their horizontal lifts, i.e., $(d\varpi)_x(X) = U$ and $(d\varpi)_x(Y) = V$, respectively. Since $\mathbf{g}_{\mathbb{CP}^{d-1}} = 4 \mathbf{g}_{\mathrm{FS}}$, the metrics $\mathbf{g}_{\mathbb{CP}^{d-1}}$ and $\mathbf{g}_{\mathrm{FS}}$ induce the same Levi--Civita connection. Hence, applying the Riemannian submersion identities with respect to $\mathbf{g}_{\mathrm{FS}}$ yields 
\begin{equation} 
	\begin{aligned} 
		(\nabla_{\mathbb{CP}^{d-1}} f)_{[x]}(U) & = (\nabla_{\mathbb{S}^{2d-1}} F)_x(X), \\ (\nabla_{\mathbb{CP}^{d-1}}^2 f)_{[x]}(U,V) & = (\nabla_{\mathbb{S}^{2d-1}}^2 F)_x(X,Y). 
	\end{aligned} 
\end{equation} 
As above, we consider the extension $\mathbb{C}^d \setminus \lbrace 0 \rbrace \ni x \mapsto \widetilde{F}(x) := F\big( \frac{x}{\vert x \vert} \big) \in \mathbb{C}$. Then, \eqref{eq:S_covar1}--\eqref{eq:S_covar2} imply
\begin{equation} 
	\begin{aligned} 
		(\nabla_{\mathbb{CP}^{d-1}} f)_{[x]}(U) & = \langle \nabla_{\mathbb{R}^{2d}}\widetilde F(x), X
		\rangle_{\mathbb{R}^{2d}}, \\
		(\nabla_{\mathbb{CP}^{d-1}}^2 f)_{[x]}(U,V) & = \langle \nabla_{\mathbb{R}^{2d}}^2\widetilde F(x) X, Y \rangle_{\mathbb{R}^{2d}}.
	\end{aligned} 
\end{equation} 
Moreover, $\Delta_{\mathbb{CP}^{d-1}} = \frac{1}{4} \Delta_{(\mathbb{CP}^{d-1},\mathbf{g}_{\mathrm{FS}})}$ as $\mathbf{g}_{\mathbb{CP}^{d-1}} = 4 \mathbf{g}_{\mathrm{FS}}$. Since $\varpi: (\mathbb{S}^{2d-1},\mathbf{g}_{\mathbb{S}^{2d-1}}) \rightarrow (\mathbb{CP}^{d-1},\mathbf{g}_{\mathrm{FS}})$ is a Riemannian submersion with totally geodesic (and thus minimal compact) fibers, it holds that
\begin{equation} 
	(\Delta_{\mathbb{CP}^{d-1}} f) \circ \varpi = \frac{1}{4} (\Delta_{(\mathbb{CP}^{d-1},\mathbf{g}_{\mathrm{FS}})} f) \circ \varpi = \frac{1}{4} \Delta_{\mathbb{S}^{2d-1}}(f \circ \varpi). 
\end{equation} 
By \eqref{eq:S_lapl}, this yields $\Delta_{\mathbb{CP}^{d-1}} f([x]) = \frac{1}{4} \Delta_{\mathbb{R}^{2d}} \widetilde F(x)$, which can be used to define the Sobolev spaces $H^s(\mathbb{CP}^{d-1})$ with $s \geq 0$ as in Section~\ref{sec:sob}.

Next, we describe the Peter--Weyl decomposition of $L^2(\mathbb{CP}^{d-1})$. Since every non-zero $Y \in \mathcal{H}_{p,q}(\mathbb{S}^{2d-1})$ satisfies $Y(e^{\mathbf{i}\theta}z) = e^{\mathbf{i}(p-q)\theta} Y(z)$, we observe that $Y$ descends to a function on $\mathbb{CP}^{d-1}$ if and only if $p=q$. Hence, the elements of $\mathcal{H}_m(\mathbb{CP}^{d-1}) := \lbrace Y \in C^\infty(\mathbb{CP}^{d-1}): Y \circ \varpi \in \mathcal{H}_{m,m}(\mathbb{S}^{2d-1}) \rbrace$ are eigenfunctions of $-\Delta_{\mathbb{CP}^{d-1}}$ with eigenvalues $\lambda_m := m(m+d-1)$ (see \cite[p.~181]{tanno79}), where $d_m := \dim(\mathcal{H}_m(\mathbb{CP}^{d-1})) = \binom{m+d-1}{d-1}^2 - \binom{m+d-2}{d-1}^2$, with the second term equal to zero if $m=0$. For every $m \in \mathbb{N}_0$, choose an ONB $(Y_{m,k})_{k=1,\ldots,d_m}$ of $\mathcal{H}_m(\mathbb{CP}^{d-1})$. Then, $(Y_{m,k})_{m \in \mathbb{N}_0, \, k=1,\ldots,d_m}$ is an ONB of $L^2(\mathbb{CP}^{d-1})$ and we obtain the Peter--Weyl decomposition
\begin{equation}
	L^2(\mathbb{CP}^{d-1}) = \widehat{\bigoplus_{m \in \mathbb{N}_0}} \mathcal{H}_m(\mathbb{CP}^{d-1}).
\end{equation}
For $s \geq 0$, we follow Section~\ref{sec:sob} and consider the Sobolev space
\begin{equation}
	H^s(\mathbb{CP}^{d-1}) := \dom\left( (1-\Delta_{\mathbb{CP}^{d-1}})^\frac{s}{2} \right) = \left\lbrace f \in \mathcal{D}'(\mathbb{CP}^{d-1}): \Vert f \Vert_{H^s(\mathbb{CP}^{d-1})} < \infty \right\rbrace,
\end{equation}
which is equipped with the norm
\begin{equation}
	\Vert f \Vert_{H^s(\mathbb{CP}^{d-1})} := \Vert (1-\Delta_{\mathbb{CP}^{d-1}})^\frac{s}{2} f \Vert_{L^2(\mathbb{CP}^{d-1})} = \left( \sum_{m=0}^\infty (1+\lambda_m)^s \sum_{k=1}^{d_m} \big\vert \widehat{f}_{m,k} \big\vert^2 \right)^\frac{1}{2},
\end{equation}
where $\widehat{f}_{m,k} := \langle f, Y_{m,k} \rangle_{L^2(\mathbb{CP}^{d-1})}$. For $K \in \mathbb{N}_0$, we define the corresponding spectral projection by
\begin{equation}
	\label{eq:defP_CP}
	L^2(\mathbb{CP}^{d-1}) \ni f \quad \mapsto \quad P_K f := \sum_{m=0}^K \sum_{k=1}^{d_m} \widehat{f}_{m,k} Y_{m,k} \in L^2(\mathbb{CP}^{d-1}),
\end{equation}
whose rank is given by $N_K := \dim(\range(P_K)) = \sum_{m=0}^K d_m = \binom{K+d-1}{d-1}^2 \asymp_d (1+K)^{2d-2}$.

Now, we approximate every physical input state $[x] \in \mathbb{CP}^{d-1}$ by a circuit-generated state. For every $[x] \in \mathbb{CP}^{d-1}$, choose $U_x \in \mathcal{PU}(d)$ with $[x] = U_x[e_1]$ such that $\mathbb{CP}^{d-1} \ni [x] \mapsto U_x \in \mathcal{PU}(d)$ is measurable, where $e_1 = \ket{0}_n$ is a reference state. Moreover, for $L \in \mathbb{N}_0$, let
\begin{equation}
	C_L([x]) \in \argmin_{V \in \mathcal{A}_{n,L}} d_{\mathcal{PU}(d)}(U_x,V),
\end{equation}
where ties are resolved by choosing the first minimizer with respect to a fixed enumeration of the finite set $\mathcal{A}_{n,L}$, ensuring that $C_L: \mathbb{CP}^{d-1} \rightarrow \mathcal{A}_{n,L} \subseteq \mathcal{PU}(d)$ is measurable. Then, $\Pi_L([x]) := C_L([x])[e_1] \in \mathbb{CP}^{d-1}$ is generated from the reference state $[e_1]$ by a quantum circuit of depth at most $L$. Let $V_L([x]) \in \mathcal{SU}(d)$ be any representative of $C_L([x])$ and set $\widetilde{\Pi}_L([x]) := V_L([x]) e_1 \in \mathbb{S}^{2d-1}$ satisfying $\Pi_L([x]) = [\widetilde{\Pi}_L([x])]$. For every $m \in \mathbb{N}_0$ and $k=1,\ldots,d_m$, the lift $Y_{m,k} \circ \varpi$ is the restriction of a bihomogeneous polynomial of bidegree $(m,m)$. Hence, there exists a matrix $B_{m,k} \in \mathbb{C}^{d^m \times d^m}$ such that for every $x \in \mathbb{S}^{2d-1}$ it holds that
\begin{equation}
	Y_{m,k}([x]) = \braket{x^{\otimes m} \vert B_{m,k} x^{\otimes m}}.
\end{equation}
Since the right-hand side is invariant under $x \mapsto e^{\mathbf{i}\theta} x$, it depends only on $[x]$. Thus, for every $K,L \in \mathbb{N}_0$, we define
\begin{equation}
	\label{eq:defQ_CP}
	\mathbb{CP}^{d-1} \ni [x] \quad \mapsto \quad (\mathcal{Q}_{K,L} f)([x]) := \sum_{m=0}^K \sum_{k=1}^{d_m} \widehat{f}_{m,k} \braket{\widetilde{\Pi}_L([x])^{\otimes m} \vert B_{m,k} \widetilde{\Pi}_L([x])^{\otimes m}} \in \mathbb{C}.
\end{equation}
Note that $(\mathcal{Q}_{K,L} f)([x])$ is independent of the chosen representative $V_L([x])$ because replacing $\widetilde{\Pi}_L([x])$ by $e^{\mathbf{i} \theta} \widetilde{\Pi}_L([x])$ leaves the inner product unchanged. Moreover, $(\mathcal{Q}_{K,L} f)([x])$ can be evaluated using quantum circuits under suitable implementation assumptions on the matrices $B_{m,k}$. Indeed, $\widetilde{\Pi}_L([x])$ is generated from $e_1 = \ket{0}_n$ by a circuit of depth at most $L$, so $\widetilde{\Pi}_L([x])^{\otimes m}$ can be prepared using $m$ parallel copies of this circuit. If the matrices $B_{m,k}$ admit efficient quantum implementations, e.g., as linear combinations of efficiently implementable unitaries, the expectation values $\braket{\widetilde{\Pi}_L([x])^{\otimes m} \vert B_{m,k} \widetilde{\Pi}_L([x])^{\otimes m}}$ can be estimated using Hadamard-test-type routines.

\begin{theorem}
	\label{thm:rate_CP}
	Let $s \geq 0$ and $f \in H^s(\mathbb{CP}^{d-1})$. Then, for every $K,L \in \mathbb{N}_0$, we have
	\begin{equation}
		\begin{aligned}
			\Vert f - \mathcal{Q}_{K,L} f \Vert_{L^2(\mathbb{CP}^{d-1})} & \lesssim_{d,s} \left( (1+K)^{-s} + (1+K)^d \varepsilon_n(L) \right) \Vert f \Vert_{H^s(\mathbb{CP}^{d-1})} \\
			&\lesssim_{d,s} \left( N_K^{-\frac{s}{2d-2}} + N_K^{\frac{1}{2}+\frac{1}{2d-2}} \varepsilon_n(L) \right) \Vert f \Vert_{H^s(\mathbb{CP}^{d-1})},
		\end{aligned}
	\end{equation}
	where $N_K \asymp_d (1+K)^{2d-2}$.
\end{theorem}
\begin{proof}
	The conclusion follows from the same arguments as in the proof of Theorem~\ref{thm:rate_S} where we now use that the orbit map $\mathcal{PU}(d) \ni [U] \mapsto [Ue_1] \in\mathbb{CP}^{d-1}$ is $2$-Lipschitz and that $\mathbb{CP}^{d-1}$ is a compact homogeneous manifold of real dimension $2d-2$.
\end{proof}

For $n=1$ with $d := 2^n = 2$, we have $\mathbb{CP}^1 \cong \mathbb{S}^2$ as Riemannian manifolds, where $\mathbb{S}^2$ denotes the Bloch sphere (see \cite[Section~5.2]{bengtsson17}). Hence, Theorem~\ref{thm:rate_CP} yields
\begin{equation}
	\label{eq:rate_CP1}
	\begin{aligned}
		\Vert f - \mathcal{Q}_{K,L} f \Vert_{L^2(\mathbb{CP}^1)} & \lesssim_s \left( (1+K)^{-s} + (1+K)^2 \varepsilon_1(L) \right) \Vert f \Vert_{H^s(\mathbb{CP}^1)} \\
		& \lesssim_s \left( N_K^{-\frac{s}{2}} + N_K \varepsilon_1(L) \right) \Vert f \Vert_{H^s(\mathbb{CP}^1)},
	\end{aligned}
\end{equation}
where $N_K \asymp (1+K)^2$.

\section{Numerical experiments}
\label{sec:numerics}

In this section, we illustrate in two numerical examples\footnote{All numerical experiments have been implemented in Python on a Lenovo ThinkPad P14s Gen6 with Processor Intel Core Ultra 7 265H, 2200 Mhz, 16 Cores. The code can be found under the following link: \url{https://github.com/psc25/QuantumHarmonic}} the approximation of functions depending on quantum unitary operations and quantum states, and examine whether the observed errors exhibit the dependence of the theoretical approximation rates.

\subsection{Function of unitary operations up to global phase}
\label{sec:num_PU2}

In the first example, we approximate a function depending on unitary operations up to global phase in $\mathcal{PU}(2) \cong \mathcal{SO}(3)$ over $1$-qubits. As target function, we consider the smooth function $f: \mathcal{PU}(2) \rightarrow \mathbb{C}$ defined by
\begin{equation}
	\mathcal{PU}(2) \ni [U] \quad \mapsto \quad f([U]) := \exp\Big( 0.9 \left( \vert\operatorname{tr}(U)\vert^2-2 \right) \!\Big) + 0.3 \, \mathbf{i} \left( \vert U_{1,1} \vert^2 - \vert U_{1,2} \vert^2 \right) \in \mathbb{C}.
\end{equation}
Both terms are invariant under $U \mapsto e^{\mathbf{i} \theta} U$, so that $f$ is well-defined on $\mathcal{PU}(2)$.

For the spectral approximation, we recall that the irreducible unitary representations $\widetilde{\pi}_\omega$ of $\mathcal{PU}(2)$ are indexed by $\omega \in 2\mathbb{N}_0$, with dimension $d_\omega = \omega+1$ and Laplace--Beltrami eigenvalue $\lambda_\omega = \frac{1}{2}\omega(\omega+2)$. Thus, for $K \in \mathbb{N}_0$, we follow \eqref{eq:defP_PU} and define
\begin{equation}
	\mathcal{PU}(2) \ni [U] \quad \mapsto \quad (P_K f)([U]) := \sum_{\omega \in 2\mathbb{N}_0 \atop \lambda_\omega \leq K^2} d_\omega \trace\big(
	\widehat{f}_\omega \widetilde{\pi}_\omega([U]) \big) \in \mathbb{C},
\end{equation}
with spectral dimension $N_K := \dim(\range(P_K)) = \sum_{\omega \in 2\mathbb{N}_0, \, \lambda_\omega \leq K^2} (\omega+1)^2$. The Fourier coefficients $\widehat{f}_\omega := \int_{\mathcal{PU}(2)} f([U]) \widetilde{\pi}_\omega([U])^* \, \mu_{\mathcal{PU}(2)}(d[U])$ are numerically approximated using the ZYZ Euler-angle parametrization, with $18$ quadrature points in each angular variable.

For the circuit approximation, we use the symmetric gate set $\mathcal{A}_1 := \lbrace H, P(\pi/4), P(\pi/4)^{-1} \rbrace$ and let $\mathcal{A}_{1,L}$ denote the corresponding projective unitaries represented by words of length at most $L$. As above, for every $[U] \in \mathcal{PU}(2)$, we choose
\begin{equation}
	Q_L([U]) \in \argmin_{[V] \in \mathcal{A}_{1,L}} d_{\mathcal{PU}(2)}([U],[V]), \qquad d_{\mathcal{PU}(2)}([U],[V]) = \sqrt{2} \arccos\left( \frac{\vert \trace(U^*V) \vert}{2} \right).
\end{equation}
Since the distance is decreasing in $\vert \operatorname{tr}(U^*V) \vert$, the minimizer is obtained by maximizing $\vert \operatorname{tr}(U^*V) \vert$. Hence, for $L \in \mathbb{N}_0$, we follow \eqref{eq:defQ_PU} and define
\begin{equation}
	\mathcal{PU}(2) \ni [U] \quad \mapsto \quad (\mathcal{Q}_{K,L} f)([U]) := (P_K f)(Q_L([U])) \in \mathbb{C},
\end{equation}
which can be implemented using $\omega$ copies of the circuit for $Q_L([U])$ (see Section~\ref{sec:rate_PU}).

For different spectral cutoffs $K \in \lbrace 0, 2, 4, 6, 8 \rbrace$ and maximal circuit depths $L \in \lbrace 8, 12, 16, 20, 24 \rbrace$, we estimate $\Vert f - \mathcal{Q}_{K,L} f \Vert_{L^2(\mathcal{PU}(2))}$ by the (empirical) root mean squared error (RMSE)
\begin{equation}
	\label{eq:RMSE_PU2}
	\operatorname{RMSE}_{K,L} := \left( \frac{1}{N} \sum_{n=1}^N \left\vert f([U_n]) - (\mathcal{Q}_{K,L} f)([U_n]) \right\vert^2 \right)^\frac{1}{2}
\end{equation}
over $N = 2000$ independent and identically distributed (i.i.d.)~samples $[U_1],\ldots,[U_N] \sim \mu_{\mathcal{PU}(2)}$. For $K > 0$, we compare the RMSE with the rate of Theorem~\ref{thm:rate_PU} (see also \eqref{eq:rate_PU2}) by fitting
\begin{equation}
	A (1+K)^{-s} + B (1+K)^\frac{5}{2} \exp\left( -a L^\frac{1}{c} \right),
\end{equation}
where $A,B,a > 0$ are parameters, $s > 0$ is the effective spectral exponent, and $c > \log_\varphi(2) \approx 1.44042$ the Solovay--Kitaev constant. Since $f$ is analytic, $s > 0$ should be interpreted as an effective finite-range spectral exponent rather than its maximal Sobolev regularity.

Figure~\ref{fig:PU} shows that the empirical RMSE exhibits a dependence on $K$ and $L$ consistent with the rate in Theorem~\ref{thm:rate_PU}, with estimated $\widehat{c} \approx 1.44052$ lying close to the imposed lower bound.

\begin{figure}[ht]
	\centering
	\includegraphics[height=6.5cm]{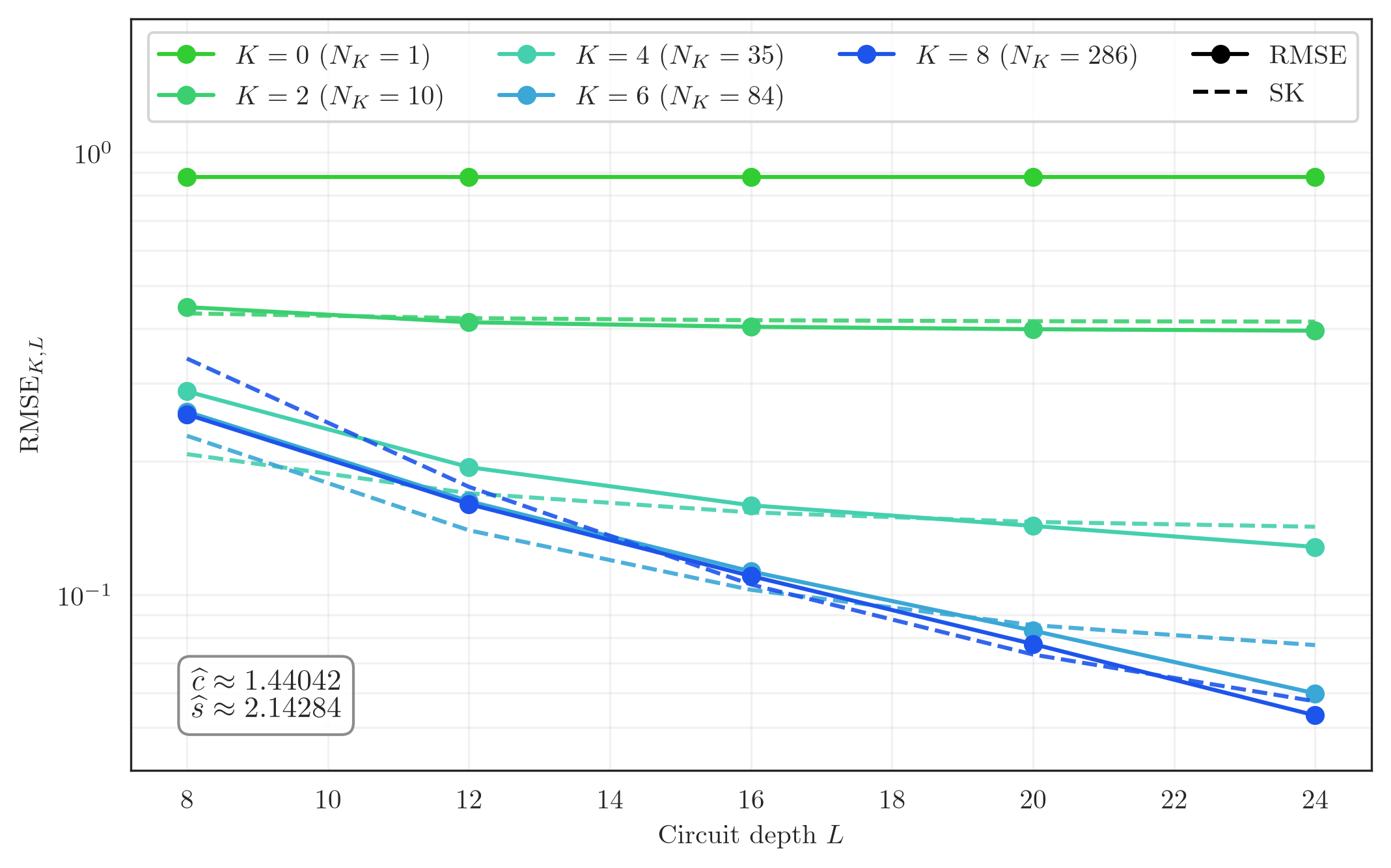}
	\caption{Empirical $\operatorname{RMSE}_{K,L}$ \eqref{eq:RMSE_PU2} between $f$ and its approximation $\mathcal{Q}_{K,L} f$ depending on the maximal circuit depth $L \in \lbrace 8, 12, 16, 20, 24 \rbrace$, for different spectral cutoffs $K \in \lbrace 0, 2, 4, 6, 8 \rbrace$. Solid lines with markers show the RMSE computed over $2000$ i.i.d.~samples $[U_1],\ldots,[U_N] \sim \mu_{\mathcal{PU}(2)}$, while dashed lines show the fitted theoretical rate from Theorem~\ref{thm:rate_PU}.}
	\label{fig:PU}
\end{figure}

\subsection{Function of physical qubit states}

In the second example, we approximate a function of physical $1$-qubit states in $\mathbb{CP}^1 \cong \mathbb{S}^2$. As target function, we consider 
\begin{equation} 
	\mathbb{CP}^1 \ni [x] \quad \mapsto \quad f([x]) := \left( 1- \left( \vert x_0 \vert^2 - \vert x_1 \vert^2 \right)^2 \right)^\frac{1+\delta}{2} \in \mathbb{R},
\end{equation} 
where $\delta = 0.05$ and $x = (x_0,x_1)^\top \in \mathbb{S}^3$ is any normalized representative of $[x]$. Note that the right-hand side is invariant under $x \mapsto e^{\mathbf{i} \theta} x$, whence $f$ is well-defined on $\mathbb{CP}^1$.

For the spectral approximation, we use the identification $\mathbb{CP}^1 \cong \mathbb{S}^2$ and recall that the spherical harmonics $(Y_{\ell,m})_{m=-\ell,\ldots,\ell}$ have Laplace--Beltrami eigenvalue $\lambda_\ell=\ell(\ell+1)$. Thus, for $K \in \mathbb{N}_0$, we follow \eqref{eq:defP_CP} and define 
\begin{equation} 
	\mathbb{CP}^1 \ni [x] \quad \mapsto \quad (P_K f)([x]) := \sum_{\ell=0}^K \sum_{m=-\ell}^{\ell} \widehat{f}_{\ell,m} Y_{\ell,m}(r([x])) \in \mathbb{C},
\end{equation}
with spectral dimension $N_K=(K+1)^2$. The Fourier coefficients are numerically approximated using spherical coordinates, with a $48$-point Gauss--Legendre quadrature in $z=\cos(\theta)$ and $96$ equidistant points in the azimuthal angle. 

For the circuit approximation, we use the same gate set $\mathcal{A}_1$ and circuit sets $\mathcal{A}_{1,L}$ as in Section~\ref{sec:num_PU2}. For a physical state $[x] \in \mathbb{CP}^1$ with Bloch angles $(\theta,\phi)$, we choose $U_x := R_z(\phi) R_y(\theta) \in \mathcal{PU}(2)$ satisfying $U_x[e_1]=[x]$. For $L \in \mathbb{N}_0$, let $C_L([x]) \in \mathcal{A}_{1,L}$ be the closest circuit to $U_x$ with respect to $d_{\mathcal{PU}(2)}$ and define $\Pi_L([x]) := C_L([x])[e_1]$. Then, we follow \eqref{eq:defQ_CP} and define
\begin{equation}
	\mathbb{CP}^1 \ni [x] \quad \mapsto \quad (\mathcal{Q}_{K,L} f)([x]) := (P_K f)(\Pi_L([x])) \in \mathbb{C}.
\end{equation}
As discussed in Section~\ref{sec:rate_CP}, the degree-$\ell$ terms in $(\mathcal{Q}_{K,L} f)([x])$ can be evaluated using $\ell$ copies of the circuit preparing $\Pi_L([x])$, together with the corresponding observable.

For spectral cutoffs $K \in \lbrace 0, 1, 2, 3, 4 \rbrace$ and maximal circuit depths $L \in \lbrace 8, 12, 16, 20, 24, 28 \rbrace$, we estimate $\Vert f - \mathcal{Q}_{K,L} f \Vert_{L^2(\mathbb{CP}^1)}$ by the (empirical) root mean squared error (RMSE) 
\begin{equation}
	\label{eq:RMSE_CP1}
	\operatorname{RMSE}_{K,L} := \left( \frac{1}{N} \sum_{n=1}^N \left\vert f([x_n]) - (\mathcal{Q}_{K,L} f)([x_n]) \right\vert^2 \right)^\frac{1}{2}
\end{equation}
over $N = 2000$ i.i.d.~samples $[x_1],\ldots,[x_N] \sim \mu_{\mathbb{CP}^1}$. As above, for $K>0$, we compare the empirical RMSE with the refined rate obtained from the eigenvalues of the Laplace--Beltrami operator $\Delta_{\mathbb{CP}^1}$, combined with the Solovay--Kitaev estimate, by fitting
\begin{equation}
	A \big( 1+(K+1)(K+2) \big)^{-\frac{s}{2}} + B(1+K)^2 \exp\left( -aL^\frac{1}{c} \right),
\end{equation}
where $A,B,a,s \geq 0$, and $c \geq \log_\varphi(2)+10^{-6} \approx 1.44052$ are fitted parameters. Figure~\ref{fig:CP} shows that the empirical RMSE exhibits a dependence on $K$ and $L$ consistent with the rate in Theorem~\ref{thm:rate_CP}, with estimated Solovay--Kitaev constant $\widehat{c} \approx 1.44052$ lying close to the imposed lower bound.

\begin{figure}[ht]
	\centering
	\includegraphics[height=6.5cm]{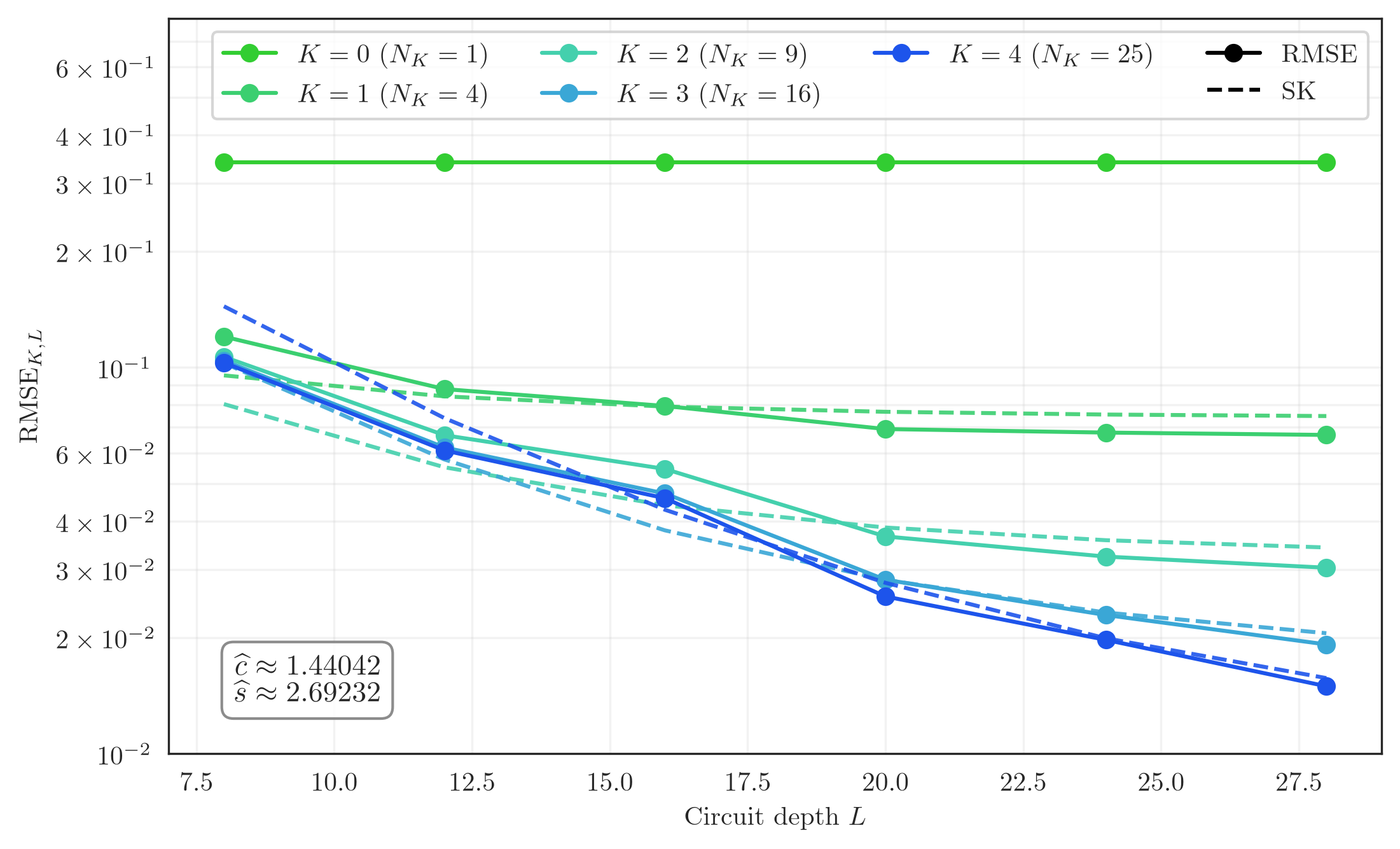}
	\caption{Empirical $\operatorname{RMSE}_{K,L}$ \eqref{eq:RMSE_CP1} between $f$ and its approximation $\mathcal{Q}_{K,L} f$ depending on the maximal circuit depth $L \in \lbrace 8, 12, 16, 20, 24, 28 \rbrace$, for different spectral cutoffs $K \in \lbrace 0, 1, 2, 3, 4 \rbrace$. Solid lines with markers show the RMSE computed over $2000$ i.i.d.~samples $[x_1],\ldots,[x_N] \sim \mu_{\mathbb{CP}^1}$, while dashed lines show the fitted theoretical rate from Theorem~\ref{thm:rate_CP}.}
	\label{fig:CP}
\end{figure}

\medskip

\noindent\textbf{Acknowledgments.} P.~Schmocker is supported by the FinsureTech Hub of ETH Zurich. 

\bibliographystyle{abbrv}
\bibliography{mybib}

\end{document}